\documentclass[a4paper]{amsart}

\usepackage[utf8]{inputenc}
\usepackage{amssymb}
\usepackage{graphicx}
\usepackage{hyperref}
\usepackage{xcolor}

\usepackage{hyperref}
\usepackage[capitalize,nameinlink,sort]{cleveref}

\usepackage{tikz}
\usetikzlibrary{arrows, automata, shapes, positioning}
\usetikzlibrary{arrows.meta}
\usetikzlibrary{decorations.pathmorphing,patterns,decorations.pathreplacing}
\usetikzlibrary{calc}
 \usepackage{float}

\usepackage{listings}

\definecolor{codepurple}{rgb}{0.58,0,0.82}
\definecolor{backcolour}{rgb}{0.95,0.95,0.92}

\lstdefinestyle{mystyle}{
    backgroundcolor=\color{backcolour},   
    commentstyle=\color{backcolour},
    keywordstyle=\color{backcolour},
    numberstyle=\tiny\color{gray},
    stringstyle=\color{codepurple},
    basicstyle=\ttfamily\footnotesize,
    breakatwhitespace=false,         
    breaklines=true,                 
    captionpos=b,                    
    keepspaces=true,                 
    numbers=left,                    
    numbersep=5pt,                  
    showspaces=false,                
    showstringspaces=false,
    showtabs=false,                  
    tabsize=2
}

\usepackage{caption}
\usepackage{subcaption}

\newcommand{\seqnum}[1]{\href{https://oeis.org/#1}{\rm \underline{#1}}}

\newcommand\NN{\mathbb{N}}
\newcommand\ZZ{\mathbb{Z}}
\newcommand\CC{\mathbb{C}}

\newcommand\val{\operatorname{val}}
\newcommand\rep{\operatorname{rep}}
\newcommand\supp{\operatorname{supp}}
\newcommand\sadd{\oplus}
\newcommand\zseq{(0)}
\newcommand\sequ{\mathbf{u}}

\newcommand\sshift{\sigma}

\newcommand{\facc}[1]{p_{\bf{#1}}}
\newcommand{\ab}[1]{\rho^\text{ab}_{\bf{#1}}}
\newcommand{\abk}[2]{\rho^{#1}_{\bf{#2}}}
\newcommand{\abeqk}[2]{\rho^{=#1}_{\bf{#2}}}
\newcommand\pav{\psi}
\newcommand\faset{\mathcal{L}}
\newcommand{\infw}[1]{{\bf{#1}}}

\newcommand{\interv}[2]{\left[#1\mathinner{{\ldotp}{\ldotp}}#2\right[}
\newcommand{\pfeq}{\operatorname{feq}}
\newcommand{\pabeq}{\operatorname{abeq}}
\newcommand{\pabexeq}{\operatorname{abexeq}}
\newcommand{\pbal}{\Delta}
\newcommand{\pzbal}{\operatorname{bal}}
\newcommand{\pocc}{\operatorname{occ}}

\theoremstyle{plain}
\newtheorem{theorem}{Theorem}
\newtheorem{lemma}[theorem]{Lemma}
\newtheorem{corollary}[theorem]{Corollary}
\newtheorem{proposition}[theorem]{Proposition}
\newtheorem{conjecture}[theorem]{Conjecture}
\theoremstyle{definition}
\newtheorem{definition}[theorem]{Definition}
\newtheorem{remark}[theorem]{Remark}

\newtheorem{example}[theorem]{Example}

\newtheorem{question}[theorem]{Question}

\numberwithin{equation}{section}

\usepackage{color}

\title[Effective Computation of Generalized Abelian Complexity]{Effective Computation of Generalized Abelian Complexity for Pisot Type Substitutive Sequences}

\author[J.-M. Couvreur et al.]{Jean-Michel Couvreur}
\address{Université d'Orléans, INSA CVL, LIFO, UR 4022, Orléans, France}
\email{jean-michel.couvreur@univ-orleans.fr}
\author[]{Martin Delacourt}
\address{Université d'Orléans, INSA CVL, LIFO, UR 4022, Orléans, France}
\email{martin.delacourt@univ-orleans.fr}{}{}
\author[]{Nicolas Ollinger}
\address{Université d'Orléans, INSA CVL, LIFO, UR 4022, Orléans, France}
\email{nicolas.ollinger@univ-orleans.fr}

\author[]{Pierre Popoli}
\address{Department of Mathematics, University of Li\`ege, Belgium}
\email{pierre.popoli@uliege.be}

\author[]{Jeffrey Shallit}
\address{School of Computer Science, University of Waterloo, Canada}
\email{shallit@waterloo.ca}

 \author[]{Manon Stipulanti}
\address{Department of Mathematics, University of Li\`ege, Belgium}
\email{m.stipulanti@uliege.be}

\begin{document}

\begin{abstract}
Generalized abelian equivalence compares words by their factors up to a certain bounded length. The associated complexity function counts the equivalence classes for factors of a given size of an infinite sequence. How practical is this notion? When can these equivalence relations and complexity functions be computed efficiently? We study the fixed points  of substitution of Pisot type.
Each of their $k$-abelian complexities is bounded and the Parikh vectors of their length-$n$ prefixes form synchronized sequences in the associated Dumont--Thomas numeration system. Therefore, the $k$-abelian complexity of Pisot substitution fixed points is automatic in the same numeration system. Two effective generic construction approaches are investigated using the \texttt{Walnut} theorem prover and are applied to several examples. We obtain new properties of the Tribonacci sequence, such as a uniform bound for its factor balancedness together with a two-dimensional linear representation of its generalized abelian complexity functions.
\end{abstract}

\maketitle

\bigskip
\hrule
\bigskip


\bigskip
\hrule
\bigskip


\bigskip
\hrule
\bigskip

A companion repository containing source code and files produced to prepare this document can be found as follows: \url{https://github.com/nopid/abcomp/}.

\bigskip
\hrule
\bigskip

\section{Introduction}
\label{sec:intro}

We consider sequences $\infw{x}$ over a finite alphabet.
One metric that has recently received some serious attention~\cite{Fici:2023} since its introduction by Richomme et al.~\cite{Richomme:2011} in 2011 is their \emph{abelian complexity}. It counts the number of distinct Parikh vectors of factors (i.e., contiguous blocks) that occur in $\infw{x}$. The \emph{Parikh vector} of a finite word records the number of occurrences of the distinct letters of the alphabet in that word. (See~\cref{sec:def} for definitions and notation.)
We deal with some generalizations of the abelian complexity, the so-called \emph{$k$-abelian complexity} (for some positive integer $k$) defined by Karhum\"aki et al.~\cite{Karhumaki:2013}.
For a positive integer $k$ and an integer $n$, the map $\abk{k}{x}(n)$ gives the number of length-$n$ factors of $\infw{x}$ that are \emph{$k$-abelian equivalent}, i.e., they share the same number of occurrences of factors of length at most $k$.

It turns out that the literature on generalized abelian complexity is limited to some famous examples.
For instance, there is a characterization of Sturmian sequences~\cite{Karhumaki:2013}.
However, computing the exact values of generalized abelian complexities is quite challenging.
Nonetheless, several papers~\cite{Chen:2018,Greinecker:2015,Lu:2022,Parreau:2015} suggest a conjecture about the inner structure of $\abk{k}{x}$ when $\infw{x}$ is produced by a finite automaton, namely, the $k$-abelian complexity of an $\ell$-automatic sequence is itself $\ell$-regular.
In this paper, we reinvestigate this conjecture and we provide two effective methods to construct a deterministic finite automaton with output (DFAO) that computes the $k$-abelian complexity of sequences satisfying some mild assumptions.
Both methods use the theorem prover {\tt Walnut}~\cite{Mousavi:2016,Shallit:2023} that relies on translating first-order logic predicates into automata and vice versa. 

The paper is organized as follows.
In~\cref{sec:def}, we introduce the setting of classical and abelian combinatorics on words, as well as the families of automatic, synchronized, and regular sequences.
In Section~\ref{sec:method 1}, we develop the first approach, which assumes that the sequence $\infw{x}$ is uniformly factor-balanced, i.e., the quantity ${| |u|_w - |v|_w | }$ is uniformly bounded for factors $u$, $v$, and $w$ of $\infw{x}$ ($u,v$ have equal length).
In this case, we show that the generalized abelian complexity of $\infw{x}$ is regular.
An innovative feature of the method, compared to previous literature, is to consider $(\abk{k}{x}(n))_{k\ge 1,n\ge 0}$ as a two-dimensional sequence.
We illustrate the effectiveness of the construction on several examples in~\cref{sec:method 1 effective computation} to~\cref{sec:Tribonacci}.
In particular, we provide new results about the well-studied Tribonacci sequence.
Then~\cref{sec:method 2} is devoted to our second method, where we consider fixed points of Pisot substitutions. (For a general discussion about Pisot type substitutions; see~\cite{Pytheas:2002}.)
First, in~\cref{sec:abelian complexity}, we obtain a DFAO computing the abelian complexity of these sequences, and as an application, we consider Parikh-collinear substitutive sequences.
Then, under a slightly different assumption, we show a similar result in~\cref{sec:k-abelian complexity} for the generalized abelian complexity.
The second method is different from the first, as we study the $k$-abelian complexity for a fixed $k$ and we translate the computation into that of the abelian complexity of the length-$k$ sliding-block code.
We note that this second method applies to a larger class of sequences.
We also illustrate it on one specific word, known as the Narayana word, in~\cref{sec:Narayana sequence}.
We finish the paper with some open questions and conjectures in~\cref{sec:open pbs}.

\section{Definitions and Notations}
\label{sec:def}

\subsection{General and abelian combinatorics on words}

Let $A^*$ denote the set of \emph{finite words} over the alphabet $A$ equipped with concatenation, and let $A^\NN$ denote the set of \emph{(infinite) sequences} over the same alphabet.
We write infinite sequences in bold.
For each $n\in \NN$, we let $A^n$ denote the set of length-$n$ words over $A$.
Let $\varepsilon$ denote the \emph{empty word}. 
For $\infw{x} \in A^*\cup A^\NN$, we let $|\infw{x}|$ denote its \emph{length}.
Let $\infw{x}[i]$ denote the letter appearing in position $0\leq i < |\infw{x}|$ inside $\infw{x}$. A \emph{factor} $u\in A^*$ of $\infw{x}$ is a sequence of consecutive letters appearing in $\infw{x}$, i.e., $u = \infw{x}[i] \cdots \infw{x}[i+n]$ for some $i,n\in\NN$. Let $\faset(\infw{x})$ denote the \emph{set of factors} of $\infw{x}$ and let $\faset_n(\infw{x}) = \faset(\infw{x})\cap A^n$ denote its the set of length-$n$ factors.
We let $\facc{x}$ denote the \emph{factor complexity} of $\infw{x}$, i.e., the map $\facc{x} \colon \NN \to \NN, n \mapsto \# \faset_n(\infw{x})$.

For a word $u\in A^*$, its \emph{Parikh vector} $\pav(u)\in\NN^A $ is defined as $\pav(u)[a] = |u|_a$ for $a\in A$, where $|u|_a$ denotes the number of occurrences of the letter $a$ in $u$. 
The \emph{abelian complexity} of a sequence $\infw{x}\in A^\NN$ is defined as $\ab{x}(n) = \#\{\pav(u)\mid u\in\faset_n(\infw{x})\}$, i.e., the number of different Parikh vectors obtained on factors of $\infw{x}$ for a given factor length~\cite{Richomme:2011}.
A generalization of the abelian complexity is the so-called \emph{$k$-abelian complexity} for some positive integer $k\ge 1$~\cite{Karhumaki:2013}.
Two words $u$ and $v$ are \emph{$k$-abelian equivalent} if $|u|_x=|v|_x$ for every word $x$ of length at most $k$, where $|w|_x$ denotes the number of occurrences of the factor $x$ in the word $w$.
We write $u \sim_k v$.
When $k=1$, we simply talk about \emph{abelian equivalence}.
For two same-length words $u,v$, we also define $u \sim_{=k} v$ if $|u|_x=|v|_x$ for every word $x$ of length exactly $k$.

It turns out that there is an equivalent definition for $k$-abelian equivalent words~\cite[Lemma~2.3]{Karhumaki:2013}.

\begin{lemma}
\label{lem:equivalent def for k-ab equiv}
Let $u,v\in A^*$ be two finite words and $k\ge 1$. The following statements are equivalent characterizations of 
$u\sim_k v$:
\begin{enumerate}
    \item The following two conditions are satisfied:
    \begin{enumerate}
    \item If $|u|<k$ or $|v|<k$, then $u=v$;
    \item Otherwise $u \sim_{=k} v$ and the length-$(k-1)$ prefixes and the length-$(k-1)$ suffixes of $u$ and $v$ are equal.
    \end{enumerate}
    \item We have $|u|=|v|$ and the following two conditions are satisfied:
    \begin{enumerate}
        \item If $|u|<k$, then $u=v$;
        \item Otherwise $u \sim_{=k} v$ and the length-$(k-1)$ prefixes of $u$ and $v$ are equal. 
    \end{enumerate}
\end{enumerate}
\end{lemma}

The \emph{$k$-abelian complexity} of a sequence $\infw{x}\in A^\NN$ is defined as $\abk{k}{x}(n) = \#\faset_n(\infw{x})/\sim_k$, i.e., we count length-$n$ factors of $\infw{x}$ up to $k$-abelian equivalence.
Similarly, we define the \emph{exact k-abelian complexity} $\abeqk{k}{x}$ of $\infw{x}$ as $\abeqk{k}{x}(n) = \#\faset_n(\infw{x})/\sim_{=k}$. 

\begin{lemma}
\label{lem:properties and observations for k-ab complexities}
Let $\infw{x}$ be a sequence.
We have $\ab{x}=\abk{1}{x}=\abeqk{1}{x}$.
For each integer $k\in\NN$, we have
\begin{enumerate}
    \item For all $n\in\NN$, $\abk{k}{x}(n) \le \abk{k+1}{x}(n)$;
    \item For all $n\in\NN$, $\abeqk{k}{x}(n) \le \abk{k}{x}(n) \le \prod_{i=1}^k \abeqk{i}{x}(n)$;
    \item For all $n<k$, $\abk{k}{x}(n)=\facc{x}(n)$. 
\end{enumerate}
\end{lemma}
\begin{proof}
The first item follows because if $u\sim_{k+1} v$ then $u\sim_k v$.
The second item is true 
as the set of words of length at most $k$ is given by $\cup_{0\le i \le k} A^i$, thus we have $u \sim_k v$ if and only if $u \sim_{= i} v$ for all $i\le k$.
The third item follows by the first item of~\cref{lem:equivalent def for k-ab equiv}.
\end{proof}

\begin{remark}
    In contrast with abelian equivalences, we do not have the implication ``$u\sim_{=(k+1)} v \Rightarrow u\sim_{=k} v$'' for all words $u,v$ and $k\ge 1$.
    For example, we have $0100\sim_{=2}1001$ but $0100 \nsim_{=1} 1001$.
    Therefore we cannot guarantee, as the first item of~\cref{lem:properties and observations for k-ab complexities}, that exact $k$-abelian complexities are increasingly nested with the same argument, i.e., we do not necessarily have $\abeqk{k}{x}(n) \le \abeqk{k+1}{x}(n)$ for all $k,n$. 
\end{remark}

There is a characterization of bounded $k$-abelian complexities, as follows.
Let $C$ be a positive integer.
A sequence $\infw{x}\in A^\NN$ is \emph{$C$-balanced} if, for all factors $u,v$ of $\infw{x}$ of equal length and for every letter $a\in A$, we have $||u|_a - |v|_a| \le C$.
When $C=1$, we usually omit the dependence on $C$, and the word is simply called \emph{balanced}.
We have the following folklore result.

\begin{lemma}
\label{lem:bounded ab iff C-bal}
A sequence $\infw{x}$ has bounded abelian complexity if and only if $\infw{x}$ is $C$-balanced for some positive integer $C$.
\end{lemma}

A generalization of $C$-balancedness is the following one.
Let $k$ and $C_k$ be positive integers.
A sequence $\infw{x}\in A^\NN$ is \emph{$(k,C_k)$-balanced} if, for all factors $u,v$ of $\infw{x}$ of equal length and for each $w\in A^k$, we have $||u|_w - |v|_w| \le C_k$.
The boundedness of the generalized abelian complexity is related to the generalized balancedness as follows.

\begin{lemma}[{\cite[Lemma~5.2.]{Karhumaki:2013}}]
\label{lem: bounded ab iff kC balanced}
Let $k$ be a positive integer.
A sequence $\infw{x}$ has bounded $k$-abelian complexity if and only if $\infw{x}$ is $(k, C_k)$-balanced for some positive integer $C_k$.
\end{lemma}

In particular, if $\abk{k}{x}$ is bounded by $C_k$, then $\infw{x}$ is $(k, C_k-1)$-balanced; conversely, if $\infw{x}$ is $(k, C_k)$-balanced, then $\abk{k}{x}\leq (C_k+1)^k$ {\cite[Lemma~5.2.]{Karhumaki:2013}}.
However, these bounds are far from being optimal in general (e.g., see~\cref{thm:Sturmian kk balanced}).

A \emph{morphism} is a map $\tau \colon A^* \rightarrow B^*$ compatible with concatenation, i.e., such that $\tau(uv) = \tau(u) \tau(v)$ for all $u,v\in A^*$. It is completely defined by its restriction $\tau_{|A} \colon A\rightarrow B^*$ to single letters. A \emph{substitution} $\tau\colon A\rightarrow A^*$ is the restriction of a morphism $\tau \colon A^* \rightarrow A^*$. A \emph{fixed point} of a substitution $\tau$ is a sequence $\infw{x}\in A^\NN$ such that $\tau(\infw{x}) = \infw{x}$. A substitution $\tau$ is \emph{prolongable} on a letter $a\in A$ if $\tau(a)=au$ for some $u\in A^*$ and $\lim_{n\rightarrow\infty} \left|\tau^n(a)\right| = +\infty$. The associated \emph{fixed point} $\tau^\omega(a)$ is $\lim_{n\rightarrow\infty} \tau^n(a) = a\prod_{n\geq 0} \tau^n(u)$. The \emph{incidence matrix} of a substitution $\tau \colon A \rightarrow A^*$ is the matrix $M_\tau\in \NN^{A\times A}$, the $(i,j)$ entry of which is $|\tau(a_i)|_{a_j}$ where $A = \{a_1,\ldots, a_n\}$.
A substitution $\tau \colon A \rightarrow A^*$ is \emph{primitive} if the corresponding matrix $M_\tau$ is primitive.

\subsection{Automatic, synchronized, and regular sequences}

An \emph{abstract numeration system with zeros (ANSZ) $\mathcal{N}$} \cite{Lecomte:2000} is a tuple $(L,A,<,0)$ where $A$ is a finite alphabet ordered by $<$ of minimal element $0\in A$ and $L\subseteq A^*$ is an infinite language of valid integer representations containing $\varepsilon$ and such that $w\in L \Leftrightarrow 0w\in L$ for all word $w\in A^*$.
The encoding $\rep_\mathcal{N}(n)$ of an integer $n\in\NN$ is the $n$th element of $L \setminus 0^+ L$ in \emph{radix order}: for all $u,v\in A^*$, let $u < v$ if $|u|<|v|$ or if $|u|=|v|$, $u\neq v$ and $u_i < v_i$ for the smallest $i$ such that $u_i\neq v_i$. The valuation $\val_\mathcal{N}(u)$ of a word $u\in L$ is $\rep_\mathcal{N}^{-1}(v)$ for the only $v\in L\setminus 0^+ L$ such that $u\in 0^* v$. 
Let $\left<.,.\right>$ denote the \emph{canonical isomorphism} between  $\cup_{n\geq 0} \left(A^n\times B^n\right)$ and  $(A\times B)^*$ for all alphabets $A, B$. 
A numeration system is \emph{regular} if both $L$ and the addition relation $\left\{ \left<x,y,z\right> \mid \val_\mathcal{N}(x) + \val_\mathcal{N}(y) = \val_\mathcal{N}(z) \right\}$ form regular languages.

A sequence $\infw{x}\in A^\NN$ is \emph{automatic in an abstract numeration system $\mathcal{N}$} (or simply \emph{$\mathcal{N}$-automatic}) if $\infw{x}$ can be computed by a DFAO in $\mathcal{N}$: the output of the DFAO on input $u\in A^*$ is defined only if $\val_\mathcal{N}(u)$ is defined and in this case it is equal to $\infw{x}[\val_\mathcal{N}(u)]$.
A sequence $s \colon \NN\rightarrow \NN^m$ form a \emph{synchronized sequence in an abstract numeration system $\mathcal{N}$} (or simply \emph{$\mathcal{N}$-synchronized}) if \[\left\{ \left<x,y_1,\ldots, y_m\right> \mid s(\val_\mathcal{N}(x)) = (\val_\mathcal{N}(y_1), \ldots, \val_\mathcal{N}(y_m)) \right\}\] is a regular language.
Finally, a sequence $\infw{x}\in A^\NN$ is \emph{regular in an abstract numeration system $\mathcal{N}$} (or simply \emph{$\mathcal{N}$-regular}) if there exist a row vector $\lambda$, a column vector $\gamma$, and a matrix-valued morphism $\mu \colon A^* \to \CC^{m\times m}$ such that $\infw{x}[n]=\lambda \mu(\rep_{\mathcal{N}}(n))\gamma$ for all $n\in\NN$.
The triple $(\lambda, \mu, \gamma)$ is called a \emph{linear representation} of $\infw{x}$. Among all linear representations computing the same function, representations of minimal dimension are called \emph{reduced representations} (sometimes called minimized in the literature).
These families of sequences are stable under several operations (e.g., sum, external product, and Hadamard product).
For more on these families of sequences, for instance see~\cite{Allouche:2000,Allouche:1992,Allouche:2003,Carpi:2001,Maes:2002,Shallit:1988,Shallit:2023}.



\section{The case of uniformly factor-balanced sequences}
\label{sec:method 1}

Our first approach to the computation of generalized abelian complexities deals with automatic sequences that are uniformly factor-balanced, namely, sequences for which the quantity ${| |u|_w - |v|_w | }$ is uniformly bounded when factors $u$, $v$ and $w$ of the sequence vary with $|u|=|v|$. Under this hypothesis, the generalized abelian equivalence predicate is synchronized and the two-dimensional generalized abelian complexity is regular. 
Since by \cref{lem: bounded ab iff kC balanced} the $k$-abelian complexity is bounded for fixed values of $k$, every $k$-abelian complexity function is also automatic.
We say that a sequence $\infw{x}\in A^\NN$ is \emph{uniformly factor-balanced} if there exists a uniform bound $C$ such that ${\left| |u|_w - |v|_w \right| \le C}$ for all $u,v\in\faset_n(\infw{x})$ for all $w\in\faset(\infw{x})$ and for all $n\in\NN$.

The factors of an automatic sequence are well captured by their appearance inside the sequence given by an index and a length. Let $\infw{x}\interv{i}{i+n}$ denote the length-$n$ factor of $\infw{x}$ starting at position $i$, i.e., $\infw{x}[i]\cdots \infw{x}[i+n-1]$. This leads to the definition of the following relations and functions:
\begin{align*}
\pfeq_\infw{x} &= \left\{(i,j,n) \mid \infw{x}\interv{i}{i+n} = \infw{x}\interv{j}{j+n}\right\};\\
\pabexeq_\infw{x} &= \left\{(i,j,k,n) \mid \infw{x}\interv{i}{i+n+k} \sim_{=k} \infw{x}\interv{j}{j+n+k}\right\};\\
\pabeq_\infw{x} &= \left\{(i,j,k,n) \mid \infw{x}\interv{i}{i+n} \sim_k \infw{x}\interv{j}{j+n}\right\};\\
\pbal_\infw{x}(i,j_1,j_2,k,n) &= \bigl|\infw{x}\interv{j_1}{j_1+n+k}\bigr|_{\infw{x}\interv{i}{i+k}} - \bigl|\infw{x}\interv{j_2}{j_2+n+k}\bigr|_{\infw{x}\interv{i}{i+k}};\\
\pzbal_\infw{x} &= \left\{(i,j_1,j_2,k,n) \mid \pbal_\infw{x}(i,j_1,j_2,k,n) = 0\right\}.
\end{align*}

\begin{lemma}\label{lemma:step balance}
The balance function $\pbal_\infw{x}(i,j_1,j_2,k,n)$ of a uniformly factor-balanced $\mathcal{N}$-automatic sequence $\infw{x}$ is $\mathcal{N}$-automatic.
\end{lemma}

\begin{proof}
Let $\infw{x}$ be $\mathcal{N}$-automatic. The relation $\pfeq_\infw{x}$ is $\mathcal{N}$-synchronized. Thus, the predicate $\pocc_\infw{x}(i,j,k,n,u)$ that tests if $j\le u\le j+n$ and $\pfeq_\infw{x}(i,j,k)$ is also $\mathcal{N}$-synchronized. Given a deterministic finite automaton (DFA) that recognizes $\pocc_\infw{x}(i,j,k,n,u)$, one can count the number of accepting paths for a given tuple $(i,j,k,n)$ to obtain a $\mathcal{N}$-regular linear representation for $\bigl|\infw{x}\interv{j}{j+n+k}\bigr|_{\infw{x}\interv{i}{i+k}}$. Combining the linear representation with itself, one obtains a linear representation for $\pbal_\infw{x}(i,j_1,j_2,k,n)$.
As $\infw{x}$ is uniformly factor-balanced, this linear representation has a finite image. Using the semigroup trick \cite[Section 4.11]{Shallit:2023}, we obtain that $\pbal_\infw{x}(i,j_1,j_2,k,n)$ is $\mathcal{N}$-automatic.
\end{proof}

As the relation $\pzbal_\infw{x}(i,j_1,j_2,k,n)$ simply tests if $\pbal_\infw{x}(i,j_1,j_2,k,n)=0$, we obtain the following result.

\begin{lemma}\label{lemma:step balsync}
If the balance function $\pbal_\infw{x}(i,j_1,j_2,k,n)$ of a sequence $\infw{x}$ is $\mathcal{N}$-automatic then its balancedness relation $\pzbal_\infw{x}(i,j_1,j_2,k,n)$ is $\mathcal{N}$-synchronized.
\end{lemma}

\begin{lemma}\label{lemma:step abcomp}
If the balancedness relation $\pzbal_\infw{x}(i,j_1,j_2,k,n)$ of a sequence $\infw{x}$ is $\mathcal{N}$-synchronized, then the associated abelian equivalence relations $\pabeq_\infw{x}(i,j,k,n)$ and $\pabexeq_\infw{x}(i,j,k,n)$ are $\mathcal{N}$-synchronized and the two-dimensional generalized abelian complexity function $(k,n)\mapsto\abk{k}{\infw{x}}(n)$ is $\mathcal{N}$-regular.
\end{lemma}

\begin{proof}
The relation $\pabexeq_\infw{x}(i,j,k,n)$ can be expressed as $\forall p\; \pzbal_\infw{x}(p,i,j,k,n)$. Following \cref{lem:equivalent def for k-ab equiv}, the relation $\pabeq_\infw{x}(i,j,k,n)$ can be expressed as a disjunction between ${\pfeq_\infw{x}(i,j,n)}$ when $n<k$ and $\pfeq_\infw{x}(i,j,k-1) \land \pabexeq_\infw{x}(i,j,k,n-k)$ when $n\ge k$. Once this relation is $\mathcal{N}$-synchronized, one can define the first occurrence of equivalent factors and from there derive a linear representation for $\abk{k}{\infw{x}}(n)$ using the path-counting technique \cite[Section 9.8]{Shallit:2023}, making the function $(k,n)\mapsto\abk{k}{\infw{x}}(n)$ $\mathcal{N}$-regular.
\end{proof}

Combining the previous lemmas we get the following theorem.

\begin{theorem}\label{thm:first approach}
Let $\infw{x}$ be a uniformly factor-balanced $\mathcal{N}$-automatic sequence.
Its abelian equivalence relation $\pabeq_\infw{x}(i,j,k,n)$ is $\mathcal{N}$-synchronized and its two-dimensional generalized abelian complexity function $(k,n)\mapsto\abk{k}{\infw{x}}(n)$ is $\mathcal{N}$-regular.
\end{theorem}

\subsection{Effective computation}
\label{sec:method 1 effective computation}

Our first approach to compute the generalized abelian complexity is quite naive and direct. It turned out to be quite computer-intensive. We were able to apply this approach only to a limited number of automatic sequences, proving a tight bound on their uniformly factor-balancedness in the process:
\begin{itemize}
    \item Some Sturmian sequences, the generalized abelian complexity of which is well known (see~\cref{thm: ab complexity of Sturmian sequences}):
    \begin{itemize}
        \item The Fibonacci sequence $\infw{f}=\varphi^\omega(0)$ where $\varphi : 0\mapsto 01,\; 1\mapsto 0$; 
        \item The Pell sequence $\infw{p}=\tau^\omega(0)$ where $\tau : 0\mapsto 001,\; 1\mapsto 0$;
    \end{itemize}
    \item The Tribonacci sequence $\infw{t}=\tau^\omega(0)$ where $\tau : 0\mapsto 01,\; 1\mapsto 02,\; 2\mapsto 0$;
    \item Some $k$-uniform fixed point $\infw{b}=\beta^\omega(0)$ where $\beta : 0\mapsto 001,\; 1\mapsto 010$.
\end{itemize}

The implementation combines several tools. The licofage toolkit \cite{licofage} was used to generate Dumont--Thomas numeration systems for fixed points of substitutions. The \texttt{Walnut} theorem prover \cite{Mousavi:2016,Shallit:2023} was used to manipulate first-order formulas and synchronized predicates. Some specific C++ programs were developed on top of the Awali \cite{Awali} library to manipulate regular sequences. In particular, the authors ported the exact rational representation of GMP \cite{GMP} to Awali and wrote an efficient OpenMP parallel reduction to reduce regular sequences in parallel. Experiments were conducted using servers with, respectively, two 24-core Intel Xeon Gold 5220R @2.2GHz processors and 64 GB of RAM, and two 24-core Intel Xeon Gold 6248R @3GHz processors and 256 GB of RAM. In both cases, with hyperthreading, 96 OpenMP threads were available to parallelize the computations.
The implementation follows the previous lemmas and is illustrated below for the Tribonacci sequence $\infw{t}$.

\subparagraph*{Implementing \cref{lemma:step balance}}
\texttt{Walnut} is first used, as follows, to produce a DFA recognizing \verb+occ_tri(i,j1,j2,k,n,u)+ ensuring that all 6 arguments are valid in the numeration system and that $\infw{t}\interv{u}{u+k} = \infw{t}\interv{i}{i+k}$ with $j_1\le u\le j_1+n$.

\begin{lstlisting}
def occ_tri "?msd_tri j1<=u & u<=j1+n & $feq_tri(i,u,k) & j2=j2":    
\end{lstlisting}

The first C++ program loads the DFA and applies the path counting argument to obtain a linear representation for $\bigl|\infw{t}\interv{j_1}{j_1+n+k}\bigr|_{\infw{t}\interv{i}{i+k}}$. Using the optimized Awali parallel code, it produces a reduced linear representation for $\pbal_\infw{t}(i,j_1,j_2,k,n)$ by computing the difference of the previous linear representation with a copy of itself where the arguments $j_1$ and $j_2$ are permuted. Then it applies the semigroup trick. If the sequence is uniformly factor-balanced, the program terminates with an automatic representation of $\pbal_\infw{t}(i,j_1,j_2,k,n)$ providing both a proof of the tightest balancedness bound and a useful DFAO computing $\pbal_\infw{t}$. This step is computer-intensive and might produce massive outputs. For the Tribonacci sequence $\infw{t}$, the computation took about 16 hours with 96 threads and produced a DFAO with 920931 states, proving that $\infw{t}$ has a tight uniform balancedness bound of 2.

\subparagraph*{Implementing \cref{lemma:step balsync}} \texttt{Walnut} is then used to define a predicate to capture the zeros of the DFAO computing $\pbal_\infw{t}(i,j_1,j_2,k,n)$. For the Tribonacci sequence, it took \texttt{Walnut} 75 seconds to compute the corresponding 487964-state DFA.

\begin{lstlisting}
def sametri "?msd_tri Dequitri[i][j1][j2][k][n] = @0":
\end{lstlisting}

\subparagraph*{Implementing \cref{lemma:step abcomp}}
\texttt{Walnut} is then used to derive the two-dimensional abelian equivalence relations and from there the first occurrence of each equivalence class.

\begin{lstlisting}
def abeqextri "?msd_tri Ai $sametri(i,j1,j2,k,n)":
def abeqtri "?msd_tri (n<k & $feq_tri(i,j,n)) 
    | (n>=k & $feq_tri(i,j,k-1) & $abeqextri(i,j,k,n-k))":
def abfirsttri "?msd_tri k>0 & ~Ej j<i & $abeqtri(i,j,k,n)":
\end{lstlisting}

The second C++ program loads the DFA and applies the path counting argument before applying the reduction algorithm to obtain a reduced linear representation for the two-dimensional generalized abelian complexity function $\abk{k}{\infw{t}}(n)$. For the Tribonacci sequence, we obtain a linear representation of dimension 264 with integer coefficients.

\subsection{Checking the validity of the result}\label{subsec:check validity}

A key element of the construction is the DFAO computing $\pbal_\infw{t}$. The validity of the DFAO can be checked inductively with first-order predicates. 
The inductive proof proceeds as follows:
\begin{enumerate}
    \item Assert that $\pbal_\infw{t}(i,j_1,j_2,k,0)$ takes only values $-1$, 0 and 1;
    \item Assert that the value of $\pbal_\infw{t}(i,j_1,j_2,k,0)$ is correct with respect to the equality of factors between $\infw{t}\interv{j_1}{j_1+k}$, respectively $\infw{t}\interv{j_2}{j_2+k}$, and $\infw{t}\interv{i}{i+k}$;
    \item Assert that $\pbal_\infw{t}(i,j_1,j_2,k,n+1) - \pbal_\infw{t}(i,j_1,j_2,k,n)$ takes only values $-1$, 0 and 1;
    \item Assert that $\pbal_\infw{t}(i,j_1,j_2,k,n+1) - \pbal_\infw{t}(i,j_1,j_2,k,n)$ is correct with respect to the equality between $\infw{t}\interv{j_1+n+1}{j_1+n+1+k}$, respectively $\infw{t}\interv{j_2+n+1}{j_2+n+1+k}$, and $\infw{t}\interv{i}{i+k}$.
\end{enumerate}

A detailed \texttt{Walnut} script is provided as an ancillary file along with the arXiv version of the paper. It took us only 45 minutes to check the 920931-state Tribonacci-DFAO.

The validity of the generalized abelian complexity has been experimentally checked against a direct approximation of the function for small values of $k$ and $n$ and against the functions computed using the second approach of~\cref{sec:method 2}.

\subsection{Application to Sturmian sequences}

Sturmian sequences are among the most famous sequences in combinatorics on words.
They have many equivalent definitions, one of which is that they are binary aperiodic sequences with minimal factor complexity, i.e., $\facc{}(n)=n+1$ (for instance, see~\cite[Chapter~2]{Lothaire:2002} for more on these sequences). In particular, with each Sturmian sequence $\infw{x}$, we associate its \emph{slope} defined by $\lim_{n \rightarrow \infty}\frac{|\infw{x}\interv{0}{n}|_1}{n}$.
The $k$-abelian complexity of Sturmian sequences is well-known and was studied in the paper~\cite{Karhumaki:2013} that introduced $k$-abelian complexities.

\begin{theorem}[{\cite[Theorem~4.1]{Karhumaki:2013}}]
\label{thm: ab complexity of Sturmian sequences}
Let $k$ be a positive integer and let $\infw{x}$ be a binary aperiodic sequence.
The sequence $\infw{x}$ is Sturmian if and only if its $k$-abelian complexity satisfies $\abk{k}{x}(n)= n+1$ if $0\le n \le 2k-1$, $\abk{k}{x}(n)= 2k$ if $n\ge 2k$.
\end{theorem}

In particular, the $k$-abelian complexity of Sturmian sequences is bounded (it is even ultimately constant).
Regarding the generalized balancedness of Sturmian sequences, we have the following result, which is more precise than~\cref{lem: bounded ab iff kC balanced}.

\begin{theorem}[{\cite[Theorem~12]{Fagnot:2002}}]
\label{thm:Sturmian kk balanced}
For any $k\ge 1$, any Sturmian sequence is $(k,k)$-balanced.    
\end{theorem}

For some classes of Sturmian sequences, we have the following result, which turns out to be finer than~\cref{thm:Sturmian kk balanced} in some cases and which is proved by putting together~\cite[Theorem~17]{Vandeth:2000} and the proof of~\cite[Corollary~13]{Fagnot:2002}.

\begin{theorem}
\label{thm:Sturmian max 3d balanced}
Let $\alpha\in(0,1)$ be a real number and $\infw{x}$ be a Sturmian sequence with slope $\alpha$. Let $\beta=\frac{\alpha}{1-\alpha}$ whose continued fraction $[b_0,b_1,b_2,\ldots]$ has bounded partial quotients.
Then, for any $k\ge 1$, the smallest integer $C_k\ge 1$ such that $\infw{x}$ is $(k,C_k)$-balanced is less than or equal to $2 + \max_i b_i$.
\end{theorem}


\begin{remark}
Recall that our assumptions require the sequence of interest to be substitutive.
It is known which Sturmian sequences are fixed points of substitutions; see~\cite{Crisp:1993} and~\cite[Section~2.3.6]{Lothaire:2002}.   
\end{remark}

\subsubsection{The Fibonacci sequence}
\label{sec:Fibonacci}

Applying \cref{thm:first approach} to the \emph{Fibonacci sequence} $\infw{f}$, the fixed point of the \emph{Fibonacci substitution} $\varphi\colon 0\mapsto 01, 1\mapsto 0$, confirms \cref{thm: ab complexity of Sturmian sequences} and provides a tight bound for its balancedness, improving on \cref{thm:Sturmian max 3d balanced} bound from 3 to 2, since the sequence $\infw{f}$ is a Sturmian sequence with slope $\alpha=\frac{3-\sqrt{5}}{2}$, giving $\beta=\frac{\sqrt{5}-1}{2}=[0,\overline{1}]$. The computation is fast and the minimal DFA for $\pbal_\infw{f}$ has only 19134 states. A careful examination of $\pbal_\infw{f}$ with \texttt{Walnut} gives the following new results (also see~\cref{sec:appendix-C}).

\begin{theorem}
\label{thm:Fib k2balanced}
Let $\infw{f}$ be the Fibonacci sequence, fixed point of the substitution $0 \mapsto 01, 1\mapsto 0$.
For each $k\geq 2$, the smallest integer $C_k\ge 1$ such that $\infw{f}$ is $(k,C_k)$-balanced is $C_k=2$.
\end{theorem}

The general balancedness of a sequence may vary for different factors. For instance, in the Fibonacci sequence $\infw{f}$, we have $||01010|_{00}-|00100|_{00}|=2$, while $||u|_w-|v|_w|\le 1$ for $w\in\{01,10\}$ and any two factors $u,v$ of the same length of $\infw{f}$.
We thus introduce the following new notion. 

\begin{definition}
Let $k,C$ be positive integers.
A sequence $\infw{x}$ is \emph{totally $(k,C)$-unbalanced} if for every length-$k$ word $w$, there exist two factors $u,v$ of $\infw{x}$ with equal length such that $||u|_w-|v|_w|>C$. 
\end{definition}

From the previous discussion, we have already established that the Fibonacci sequence $\infw{f}$ is neither $(2,1)$-balanced nor totally $(2,1)$-unbalanced. Indeed, with \texttt{Walnut}, we determine that $\infw{f}$ is $(2,1)$-balanced on $01,10$ but not on $00$. Moreover, we also prove that $\infw{f}$ is not $(3,1)$-balanced nor totally $(3,1)$-unbalanced. 
More generally, we may prove the following (see~\cref{sec:appendix-C}).

\begin{theorem}
\label{thm: balanced Fib}
Let $\infw{f}$ be the Fibonacci sequence, fixed point of the substitution $0 \mapsto 01, 1\mapsto 0$.
For each $k\ge 4$, and only for these, $\infw{f}$ is totally $(k,1)$-unbalanced.
\end{theorem}

\subsubsection{The Pell sequence}
\label{sec:Pell seq}

Applying \cref{thm:first approach} to the \emph{Pell sequence} $\infw{p}$, indexed~\cite[\seqnum{A171588}]{Sloane}, the fixed point of the \emph{Pell substitution} $\tau\colon 0\mapsto 001, 1\mapsto 0$, confirms \cref{thm: ab complexity of Sturmian sequences} and provides a tight bound for its balancedness, improving on \cref{thm:Sturmian max 3d balanced} bound from 4 to 3, since the sequence $\infw{p}$ is a Sturmian sequence with slope $\alpha=\frac{2-\sqrt{2}}{2}$, giving $\beta=\sqrt{2}-1=[0,\overline{2}]$.  The computation is fast and the minimal DFA for $\pbal_\infw{p}$ has only 28713 states. A careful examination of $\pbal_\infw{p}$ with \texttt{Walnut} gives the following new results (also see~\cref{sec:appendix-C}/supplementary material of the paper).

\begin{theorem}
\label{thm:Pell k3balanced}
Let $\infw{p}$ be the Pell sequence, fixed point of the substitution $\tau\colon 0\mapsto 001, 1\mapsto 0$.
For each $k\geq 10$, the smallest integer $C_k\ge 1$ such that $\infw{p}$ is $(k,C_k)$-balanced is $C_k=3$.
\end{theorem}

\begin{theorem}
\label{thm: unbalanced Pell}
Let $\infw{p}$ be the Pell sequence, fixed point of the substitution $0 \mapsto 001, 1\mapsto 0$.
For each $k\ge 6$, and only for these, $\infw{p}$ is totally $(k,1)$-unbalanced. It is not totally $(k,2)$-unbalanced for any $k\in\NN$.
\end{theorem}




\subsection{Application to the Tribonacci sequence}
\label{sec:Tribonacci}

We consider the \emph{Tribonacci sequence} $\infw{t}=010201001\cdots$, the fixed point starting with $0$ of the \emph{Tribonacci substitution} $\tau\colon 0\mapsto 01, 1\mapsto 02, 2 \mapsto 0$. This sequence satisfies $\facc{\infw{t}}(n)=2n+1$ and is called an \emph{episturmian sequence}~\cite{Glen-2009}, a well-known generalization of Sturmian words.
Applying \cref{thm:first approach} to this sequence provides several new results for this well-studied sequence.

\begin{theorem}
Let $\infw{t}$ be the Tribonacci sequence, fixed point of $0\mapsto 01, 1\mapsto 02, 2 \mapsto 0$.
The two-dimensional generalized abelian complexity function $(k,n)\mapsto \abk{k}{\infw{t}}(n)$ is regular in the Tribonacci numeration system. It admits a reduced linear representation of dimension 264.
\end{theorem}

\begin{theorem}
Let $\infw{t}$ be the Tribonacci sequence, fixed point of $0\mapsto 01, 1\mapsto 02, 2 \mapsto 0$.
Both first difference functions $(k,n)\mapsto \Delta_k\abk{k}{\infw{t}}(n) = \abk{k+1}{\infw{t}}(n)-\abk{k}{\infw{t}}(n)$ and  $(k,n)\mapsto \Delta_n\abk{k}{\infw{t}}(n) =\abk{k}{\infw{t}}(n+1)-\abk{k}{\infw{t}}(n)$ are automatic in the Tribonacci numeration system and thus bounded. \cref{fig:tribo differences} depicts the first values of these functions.
\end{theorem}

\begin{corollary}
Let $\infw{t}$ be the Tribonacci sequence, fixed point of $0\mapsto 01, 1\mapsto 02, 2 \mapsto 0$.
For each $k\ge 1$, the $k$-abelian complexity function $n \mapsto \abk{k}{\infw{t}}(n)$ is automatic in the Tribonacci numeration system and can be constructed efficiently and recursively on $k$.
\end{corollary}

It is well known that the Tribonacci sequence $\infw{t}$ is $(1,2)$-balanced (see \cite{Richomme:2010}), and the proof of this result is nontrivial. It is also known that the sequence is $(k,C_k)$-balanced for all $k\ge 2$  (see \cite{Berthe:2019}), but to our knowledge, no precise bound was hitherto known. A careful analysis of the 920931-state DFAO of $\pbal_\infw{t}$ provides a uniform bound (also see~\cref{sec:appendix-C}).

\begin{theorem}
\label{thm:Trib k2balanced}
Let $\infw{t}$ be the Tribonacci sequence, fixed point of $0\mapsto 01, 1\mapsto 02, 2 \mapsto 0$.
For each $k\geq 1$, the smallest integer $C_k\ge 1$ such that $\infw{t}$ is $(k,C_k)$-balanced is $C_k=2$.
\end{theorem}

\begin{theorem}
\label{thm:Trib k2unbalanced}
Let $\infw{t}$ be the Tribonacci sequence, fixed point of $0\mapsto 01, 1\mapsto 02, 2 \mapsto 0$.
For each $k\geq 1$, the sequence $\infw{t}$ is totally $(k,1)$-unbalanced.
\end{theorem}

\begin{figure}
    \centering
    \begin{subfigure}[t]{0.4\textwidth}
    \includegraphics[width=\textwidth]{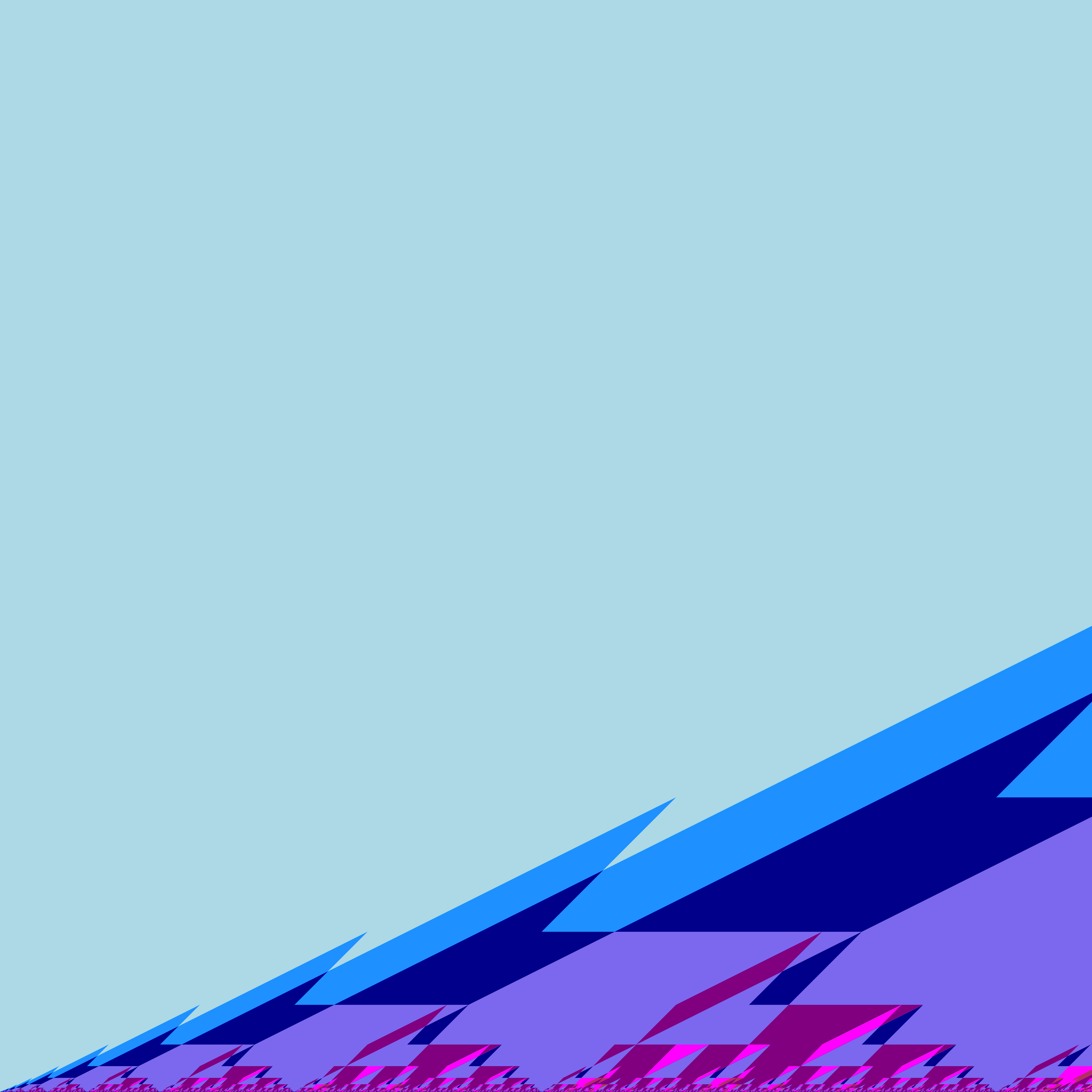}
    \caption{Two-dimensional automatic $\Delta_k\abk{k}{\infw{t}}(n)$.}        
    \end{subfigure}
    \hfil
    \begin{subfigure}[t]{0.4\textwidth}
    \includegraphics[width=\textwidth]{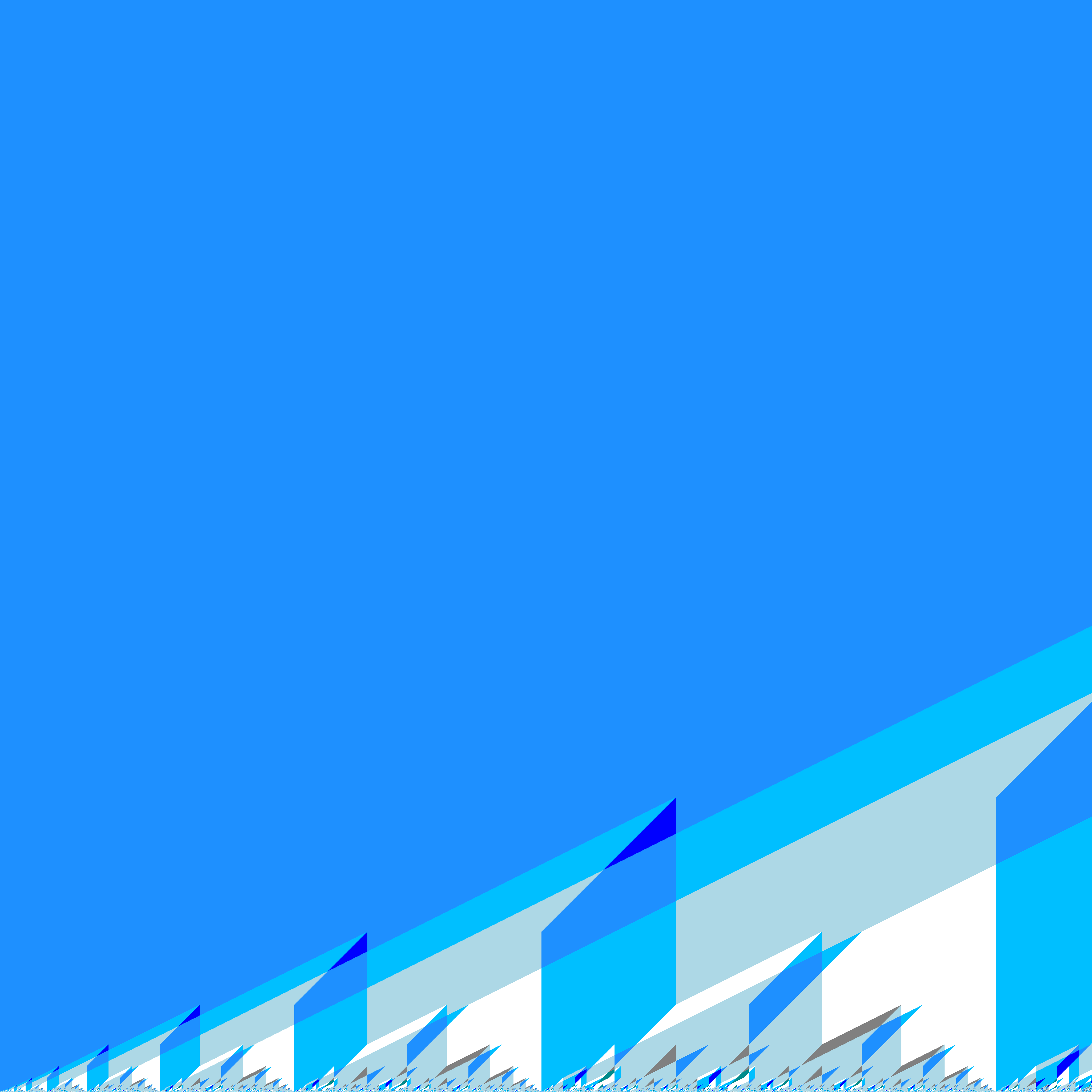}
    \caption{Two-dimensional automatic $\Delta_n\abk{k}{\infw{t}}(n)$.}        
    \end{subfigure}
    \caption{The Tribonacci sequence has bounded discrete derivatives (first 4096 values, $n$ on the horizontal axis, $k$ on the vertical axis, one color per value, a different palette for each picture).}
    \label{fig:tribo differences}
\end{figure}

\section{The case of Pisot substitutions}
\label{sec:method 2}

In this section, we handle the computation of the $k$-abelian complexity of fixed points of substitutions satisfying an assumption different from the uniform balancedness of the previous section. This second approach relies on ultimately Pisot substitutions and makes use of the concept of so-called sequence automata introduced by Carton et al.~\cite{Carton:2024}. 

A \emph{Pisot-Vijayaraghavan number} $\theta$ is an algebraic integer, which is the dominant root of its minimal monic polynomial $P(X)$ with integer coefficients, where $P(X)$ is irreducible over $\ZZ$ and admits $n$ complex roots $\theta_1$, \ldots, $\theta_n$, all distinct, satisfying $\theta=\theta_1 > 1 > |\theta_2| \geq \cdots \geq |\theta_n| > 0$. A substitution is of \emph{Pisot type}, or simply \emph{Pisot}, if the characteristic polynomial of its incidence matrix is the minimal polynomial of a Pisot number.
A substitution is of \emph{ultimately Pisot type}, or simply \emph{ultimately Pisot}, if the characteristic polynomial of its incidence matrix is the minimal polynomial of a Pisot number $\theta$ multiplied by a power of $X$, i.e., $X^m \cdot P_\theta(X)$ for some $m\ge 0$.
Such a polynomial is called \emph{ultimately Pisot}.
(Other terms exist to designate this property, e.g., \emph{Pisot up to a shift} in~\cite{Carton:2024}.)
Combining~\cite[Lemma~2.2]{Richomme:2011} and~\cite[Theorem~13]{Adamczewski:2003} gives the next result. 

\begin{theorem}\label{thm:bounded ab compl}
    The abelian complexity of the fixed point of a prolongable primitive substitution of ultimately Pisot type is bounded by a constant.
\end{theorem}

The \emph{addressing automaton} $\mathcal{A}_\varphi$ associated with the fixed point $\varphi^\omega(a)$ of a prolongable substitution $\varphi \colon A \rightarrow A^*$ is the DFA with state set~$A$, alphabet $\{0,1,\ldots,n-1\}$ where $n=\max_{b\in A}|\varphi(b)|$, initial state~$a$, final states~$A$ and whose transitions are defined by $\varphi$ as $\delta(b,i) = \varphi(b)_i$ for all $b\in A$ and $i \in \{0,\ldots,|\varphi(b)|-1\}$.
Let $L_\varphi = L(\mathcal{A}_\varphi)\setminus 0^+A^*$ where $L(\mathcal{A}_\varphi)$ is the language recognized by $\mathcal{A}_\varphi$. The \emph{Dumont--Thomas numeration system $\mathcal{N}_\varphi$} associated with the fixed point $\varphi^\omega(a)$ is the ANSZ $(L_\varphi, \{0,1,\ldots,n-1\}, <, 0)$ where $<$ is the usual order on $\NN$. The fixed point $\varphi^\omega(a)$ is automatic for $\mathcal{N}_\varphi$: the addressing automaton provides a valid DFAO when equipped with the output function $\pi \colon q\mapsto q$. See~\cite{Carton:2024} for more details.

\begin{theorem}[Carton et al.~\cite{Carton:2024}]
\label{thm:dtpisot}
    The Dumont--Thomas numeration system associated with a fixed point of a prolongable substitution of ultimately Pisot type is regular.
\end{theorem}

 Let $\sequ \in \ZZ^\NN$ be an \emph{integer sequence} and let $\zseq$ denote the constant sequence everywhere equal to $0$. The \emph{shift operator} $\sshift \colon \ZZ^\NN \rightarrow \ZZ^\NN$ removes the first element of a sequence, i.e., $(\sshift\sequ)[n] = \sequ[n+1]$ for all $\sequ\in\ZZ^\NN$ and $n\in\NN$.

A \emph{sequence automaton} is a \emph{partial DFA} $(Q, A, \delta, q_0, F)$ equipped with a \emph{partial vector map} $\pi \colon Q\times A \rightarrow \ZZ^\NN$. Formally, $Q$ is the \emph{finite set of states}, $A$ is the \emph{finite alphabet of symbols}, $\delta \colon Q\times A\rightarrow Q$ is the \emph{partial transition map}, $q_0\in Q$ is the \emph{initial state}, $F\subseteq Q$ is the set of \emph{accepting states}. The transition and the vector map share the same domain. The transition map and vector map are inductively extended from symbols to words as follows, for all $q\in Q$, $u\in A^*$ and $a\in A$:
\begin{align*}
\delta(q,\varepsilon) &= q,  &\pi(q,\varepsilon) &= \zseq,\\
\delta(q, ua) &= \delta(\delta(q, u), a), &\pi(q, ua) &= \sshift\pi(q, u) + \pi(\delta(q, u), a).
\end{align*}

The class of \emph{$\ZZ$-rational series} is a well-studied class of functions from finite words to $\ZZ$. For a general introduction, see~\cite{Berstel:2011}. It admits finite linear representations.
Similarly to regular sequences, $\ZZ$-rational series are closed under several operations.
In particular, they are also closed under \emph{synchronized addition} $f\oplus g \colon \left<u,v\right> \mapsto f(u) + g(v)$ for all $\ZZ$-rational series $f \colon A^* \rightarrow \ZZ$ and $g \colon B^* \rightarrow \ZZ$. The \emph{support} $\supp(f)$ of a rational series $f$ is the language $A^* \setminus f^{-1}(0)$.

Let the \emph{series} $s_\mathcal{A}$ of a sequence automaton $\mathcal{A}$ map every word $u\in A^*$  to the first element of its vector $\pi(q_0, u)[0]$ when defined, or to $0$ otherwise. Let a sequence automaton be \emph{linear recurrence} when all sequences in the vector map of a sequence automaton are \emph{linear recurrence sequences}. In this case, the \emph{series} $s_\mathcal{A}$ is $\ZZ$-rational. The \emph{recurrence polynomial} $P_\mathcal{A}$ of a linear recurrence sequence automaton (LRSA) $\mathcal{A}$ is the minimal polynomial for all the sequences in the image of the vector map. 

Using the product of DFA and linear combinations of vector maps, sequence automata can be combined to produce linear combinations of sequence automata. For every sequence automata $\mathcal{A}$, $\mathcal{B}$ and $\alpha\in\ZZ$, let $\mathcal{A} + \mathcal{B}$ denote the \emph{sum sequence automaton} with series $s_\mathcal{A} \sadd s_\mathcal{B}$ and let $\alpha\mathcal{A}$ denote the \emph{external product sequence automaton} with series $\alpha s_\mathcal{A}$. If $\mathcal{A}$ and $\mathcal{B}$ are LRSA then $P_{\mathcal{A}+\mathcal{B}}$ divides the least common multiple of $P_\mathcal{A}$ and $P_\mathcal{B}$ and $P_{\alpha\mathcal{A}}$ divides $P_\mathcal{A}$.

Let $\mathcal{A}_\varphi$ be the addressing automaton of a Dumont--Thomas numeration system $\mathcal{N}_\varphi$ associated with the fixed point $\varphi^\omega(a)$ of a prolongable substitution $\varphi \colon A \rightarrow A^*$. The \emph{addressing sequence automaton} $\mathcal{S}_\varphi$ of $\mathcal{N}_\varphi$ is the LRSA derived from $\mathcal{A}_\varphi$ with the vector map $
\pi(a, i) = \left( |\varphi^n\left(\varphi(a)[0]\cdots\varphi(a)[i-1]\right)| \right)_{n\in\NN}$
for all $a\in A$ and  $i\in\{0,\dots,|\varphi(a)|-1\}$. 
By the Cayley--Hamilton theorem, its recurrence polynomial $P_{\mathcal{S}_\varphi}$ divides the characteristic polynomial of the incidence matrix of $\varphi$.
The series of $\mathcal{S}_\varphi$ is the \emph{valuation series} $\nu_\varphi$ of the Dumont--Thomas numeration system $\nu_\varphi(u) = \val_{\mathcal{N}_\varphi}(u) \text{~if defined,~} 0 \text{~otherwise}$. The numeration system $\mathcal{N}_\varphi$ is regular if and only if $\supp(\nu_\varphi \sadd \nu_\varphi \sadd -\nu_\varphi)$ is regular, i.e.,  if the support of the series of the LRSA $\mathcal{S}_\varphi + \mathcal{S}_\varphi - \mathcal{S}_\varphi$ is regular.
~\cref{thm:dtpisot} is the corollary of the following more technical proposition. 

\begin{proposition}[Carton et al.~\cite{Carton:2024}]\label{prop:supp}
   The support of the series of a  LRSA with ultimately Pisot recurrence polynomial is regular.
\end{proposition}

When two Dumont--Thomas numeration systems $\mathcal{N}_\varphi$ and $\mathcal{N}'_\varphi$ are associated with the same Pisot number, $\supp(\nu_\varphi \sadd -\nu_{\varphi'})$ is regular and thus by \cref{prop:supp} the \emph{converter} ${\{ \left<u,v\right> \mid \val_{\mathcal{N}_\varphi}(u) = \val_{\mathcal{N}_{\varphi'}}(v) \}}$ between $\mathcal{N}_\varphi$ and $\mathcal{N}_{\varphi'}$ is regular.

\begin{proposition}[Carton et al.~\cite{Carton:2024}]\label{prop:converter}
The converter between two Dumont--Thomas numeration systems associated with a common Pisot number is regular.    
\end{proposition}

Parikh vectors of length-$n$ prefixes of the fixed point of a prolongable substitution can be obtained by a slight modification of the addressing sequence automaton.  The \emph{Parikh sequence automaton} $\mathcal{S}^b_\varphi$ for $b\in A$ is the LRSA derived from $\mathcal{A}_\varphi$ with the vector map $
\pi_b(a, i) = \left( |\varphi^n\left(\varphi(a)[0]\cdots\varphi(a)[i-1]\right)|_b \right)_{n\in\NN}$
for all $a\in A$ and  ${i\in\{0,\dots,|\varphi(a)|-1\}}$. 
By the Cayley--Hamilton theorem, its recurrence polynomial $P_{\mathcal{S}^b_\varphi}$ divides the characteristic polynomial of the incidence matrix of $\varphi$. The series of $\mathcal{S}^b_\varphi$ is the \emph{Parikh prefix series}
\[
\nu^b_\varphi(u) = \left|\varphi^\omega(a)[0]\cdots\varphi^\omega(a)[\val_{\mathcal{N}_\varphi}(u)-1]\right|_b \text{~if defined,~} 0 \text{~otherwise.}
\]

\begin{proposition}[Carton et al. \cite{Carton:2024}]\label{prop:synch}
   Let $\infw{x}$ be a fixed point of a prolongable substitution $\varphi$.
   The Parikh vectors of length-$n$ prefixes of $\infw{x}$ form a synchronized sequence when the support of $\mathcal{S}^b_\varphi - \mathcal{S}_\varphi$ is regular for all $b\in A$.
\end{proposition}

\subsection{Abelian complexity}
\label{sec:abelian complexity}

Let us first recall the recent result of Shallit~\cite{Shallit:2021} about the abelian complexity of an automatic sequence $\infw{x}$ under the assumption that the Parikh vectors of length-$n$ prefixes of $\infw{x}$ are synchronized. 

\begin{theorem}[Shallit \cite{Shallit:2021}]\label{thm:shal}
    Let $\infw{x}\in A^\NN$ be automatic in some regular numeration system $\mathcal{N}$. Suppose that
    \begin{enumerate}
        \item The abelian complexity $\ab{x}(n)$ is bounded above by a constant, and
        \item The Parikh vectors of length-$n$ prefixes of $\infw{x}$ form an $\mathcal{N}$-synchronized sequence.
        \end{enumerate}
    Then $\ab{x}(n)$ is an $\mathcal{N}$-automatic sequence and the DFAO computing it is effectively computable.
\end{theorem}

We can now tackle the abelian complexity of fixed point of primitive substitution of ultimately Pisot type. 

\begin{theorem}
\label{thm: ab complexity Pisot substitution is automatic in DT}
  The abelian complexity of the fixed point of a prolongable primitive substitution of ultimately Pisot type is an automatic sequence in the associated Dumont--Thomas numeration system and the DFAO computing it is effectively computable.
\end{theorem}

\begin{proof}
    Let $\infw{x}$ be the fixed point $\tau^\omega(a)$ of a prolongable primitive substitution $\tau \colon A \rightarrow A^*$ of ultimately Pisot type. The characteristic polynomial $P(X)$ of the incidence matrix of $\tau$ is of the form $P(X)=X^kQ(X)$ for some $k\geq 0$ and some minimal polynomial $Q$ of a Pisot number. Therefore, the LRSA $\mathcal{S}_\tau$ and $\mathcal{S}^b_\tau$ all have the same ultimately Pisot recurrence polynomial $P(X)$. Let $\mathcal{N}_\tau$ be the associated Dumont--Thomas numeration system. To conclude, we want to ensure we can apply~\cref{thm:shal}:
\begin{enumerate}
\item The sequence $\infw{x}$ is indeed automatic in $\mathcal{N}_\tau$.
\item The numeration system $\mathcal{N}_\tau$ is regular by~\cref{thm:dtpisot}.
\item The abelian complexity of $\infw{x}$ is bounded by~\cref{thm:bounded ab compl}. 
\item By~\cref{prop:synch}, the Parikh vectors of length-$n$ prefixes of $\infw{x}$ form a synchronized sequence as the support of $\mathcal{S}^b_\tau - \mathcal{S}_\tau$ is regular for all $b\in A$ by~\cref{prop:supp}.
\end{enumerate}
The construction of every step is effective, thus the DFAO can be effectively computed.
\end{proof}

A more conventional numeration system associated with a Pisot root is its canonical Bertrand numeration system \cite{Bruyere:1997,Frougny:1996,Charlier:2022}. It can be described as a particular Dumont--Thomas numeration system that shares the same Pisot recurrence polynomial. As a consequence, the conversion between both numeration systems can be realized using LRSA with regular support by~\cref{prop:converter}.

\begin{corollary}
\label{cor: ab complexity Pisot substitution is automatic in Bertrand}
   The abelian complexity of the fixed point of a prolongable primitive substitution of ultimately Pisot type is an automatic sequence in the associated canonical Bertrand numeration system and the DFAO computing it is effectively computable.
\end{corollary}

As an application of the previous two results, we turn to so-called Parikh-collinear morphisms.
A morphism $\tau \colon A^* \rightarrow B^*$ is \emph{Parikh-collinear} if the Parikh vectors $\Psi(\tau(a))$, $a\in A$, are collinear.
In other words, Parikh-collinear morphisms have an incidence matrix of rank $1$ (unless they are completely erasing).
This family of morphisms has gained some scientific interest over the years, for instance, see~\cite{Allouche:2021,Cassaigne:2011,Popoli:2024,Rigo:2023,Rigo:2024,Whiteland:2021}. 
With the next result, we note that we recover a particular case of~\cite[Theorem~3]{Rigo:2024}.

\begin{corollary}
\label{cor: ab complexity Parikh collinear}
Let $\tau \colon A \rightarrow A^*$ be a Parikh-collinear prolongable primitive substitution with fixed point $\infw{x}$.
Define $\alpha = \sum_{a\in A} |\tau(a)|_a$.
The abelian complexity of $\infw{x}$ is automatic in both the associated Dumont--Thomas numeration system and in base $\alpha$; moreover, the DFAOs generating it are effectively computable.
\end{corollary}
\begin{proof}
The characteristic polynomial of $\tau$ is $X^\ell(X-\alpha)$ where $\ell= |A|-1$.
In particular, $\tau$ is ultimately Pisot.
Now the result follows from either~\cref{thm: ab complexity Pisot substitution is automatic in DT} or~\cref{cor: ab complexity Pisot substitution is automatic in Bertrand} since the canonical Bertrand numeration system for $\tau$ is the classical integer base $\alpha$.
\end{proof}

\begin{example}
\label{ex: z1 Parikh-collinear}
Consider the fixed point $\infw{z}=0100111001\cdots$ of the Parikh-collinear primitive substitution $0 \mapsto 010011, 1\mapsto 1001$.
The abelian complexity of $\infw{z}$ is aperiodic~\cite[Proposition~13]{Rigo:2023} and there is a base-$5$ DFAO with $9$ states that computes $\abk{1}{z}$,  see~\cite[Section~5.2]{Rigo:2023}.
It is interesting to note that~\cref{thm: ab complexity Pisot substitution is automatic in DT} gives a $15$-state DFAO in the corresponding Dumont--Thomas numeration system.
Our procedure also allows to convert from base-$5$ to this numeration system and vice versa.
\end{example}

\subsection{Generalized abelian complexity}
\label{sec:k-abelian complexity}

In this section, we turn to the more generalized notion of $k$-abelian complexity of sequences.
We note that the assumptions of the main result of this section slightly differ from those of the previous section.
In short, the method is to translate the problem by studying the abelian complexity of the length-$k$ sliding-block code.
For an illustration of the concepts of this section, see~\cref{sec:appendix-B}.

\begin{definition}[Sliding-block code]
\label{def: sliding block code}
For a sequence $\infw{x}$ and each integer $k\ge 1$, we let $B_k(\infw{x})$ denote the \emph{length-$k$ sliding-block code of $\infw{x}$}, i.e., if $\infw{x} = x_0 x_1 x_2 \cdots$, then we slide a length-$k$ window in $\infw{x}$ to group length-$k$ factors
\[
(x_0 x_1 \cdots x_{k-1}) (x_1 x_2 \cdots x_k) \cdots (x_i x_{i+1} \cdots x_{i+k-1}) \cdots,
\]
and we map distinct length-$k$ factors to distinct letters in a new alphabet of size $\#\faset_k(\infw{x})$ to code $B_k(\infw{x})$.
\end{definition}

It is worth noticing that the letter $i$ in the length-$k$ sliding-block code corresponds to the $i$th factor of length $k$ appearing in the original sequence. 


\begin{lemma}
\label{lem:=k ab equal ab of k-block code}
For a sequence $\infw{x}$ and each integer $k\ge 1$, we have $\abeqk{k}{x}(n+k-1) = \rho_{B_k(\infw{x})}^1(n)$ for all $n\in\NN$.
\end{lemma}
\begin{proof}
To compute $\abeqk{k}{x}(n+k-1)$ (resp., $\rho_{B_k(\infw{x})}^1(n)$), we need to count $\faset_{n+k-1}(\infw{x})/_{\sim_{=k}}$ (resp., $\faset_{n}(B_k(\infw{x}))/_{\sim_1}$).
To conclude, observe that $\faset_{n}(B_k(\infw{x}))$ and $\faset_{n+k-1}(\infw{x})$ are in bijection.
\end{proof}

Let $\tau\colon A \to A^*$ be a substitution prolongable on $a\in A$.
Let $\infw{x}=\tau^\omega(a)$ be the fixed point of $\tau$ with starting letter $a$.
Let $M_\tau$ be the incidence matrix of $\tau$ and let $P_\tau$ be the characteristic polynomial of $M_\tau$.
We note that $P_\tau$ is also monic.
Let $\mathcal{N}_\tau$ be the associated Dumont--Thomas numeration system. 
Then $\infw{x}$ is $\mathcal{N}_\tau$-automatic.
If $\tau$ is ultimately Pisot type, then~\cref{thm: ab complexity Pisot substitution is automatic in DT} assures that the abelian complexity of $\infw{x}$ is $\mathcal{N}_\tau$-automatic.

Let $\tau_k$ denote the substitution derived from $\tau$ such that $B_k(\infw{x})$ is its fixed point.
More precisely, we define $A_k=\{1,\ldots,\facc{x}(k)\}$ and $\Theta_k \colon \faset_k(\infw{x}) \to A_k$ to encode the order in which the length-$k$ factors of $\infw{x}$ appear in $\infw{x}$.
To define $\tau_k\colon A_k \to A_k^*$, for each $\ell\in A_k$, the word $\tau_k(\ell)$ consists of the ordered list of the first $|\tau(u[0])|$ length-$k$ factors of $\tau(u)$, where $u=\Theta^{-1}_k(\ell)$ (i.e., $u$ is the $\ell$th length-$k$ factor encountered in $\infw{x}$).
See~\cite[Section~5.4]{Queffelec:2010} for more details.

Using~\cite[Section~5.4.3]{Queffelec:2010} (also see~\cite[Proposition~21]{Adamczewski:2003}), if $\tau$ is primitive, then $\tau_k$ is also primitive and the dominant Perron eigenvalue of $M_{\tau_k}$ is that of $M_\tau$.
Moreover, the eigenvalues of $M_{\tau_k}$ (with $k\ge 2$) are those of $M_{\tau_2}$ with additional zeroes~\cite[Corollary~5.5]{Queffelec:2010}, i.e., $P_{\tau_k}(X) = X^m P_{\tau_2}(X)$ for some integer $m\ge 0$.
This identity implies the next result.

\begin{lemma}
If $\tau_2$ is Pisot, then $\tau_k$ is ultimately Pisot with the same Pisot root for $k\ge 2$.  
\end{lemma}

\begin{lemma}
\label{lem:subset of spectres}
For each $k\ge 2$, each eigenvalue of $M_\tau$ is also an eigenvalue of $M_{\tau_k}$.
\end{lemma}
\begin{proof}
Fix some integer $k\ge 2$ and recall that $A=\{0,1,\ldots,n-1\}$.
Define $\pi_k \colon A_k \to A, i \mapsto (\Theta_k^{-1}(i))[0]$, i.e., $\pi_k(i)$ encodes the first letter of the $i$th length-$k$ factor encountered in $\infw{x}$.
Now let $V$ be an eigenvector of $M_\tau$ with eigenvalue $\alpha$, i.e., $M_\tau V = \alpha V$.
Define the vector $V_k$ such that its $i$th component is given by $V_k[i]=V[\pi_k(i)]$.
Then we show that $M_{\tau_k} V_k = \alpha V_k$.
Fix $i\in\{1,\ldots,\facc{x}(k)\}$.
We have
\begin{align*}
(M_{\tau_k} V_k)[i]
&= \sum_{\ell=1}^{\facc{x}(k)} M_{\tau_k}[i,\ell] V_k[\ell] = \sum_{m=1}^n \left( \sum_{j \in \pi_k^{-1}(m)} M_{\tau_k}[i,j] \right) V[m],\\
&= \sum_{m=1}^n M_\tau [\pi_k(i),m] V[m]= (M_\tau V) [\pi_k(i)]= \alpha V[\pi_k(i)]= \alpha V_k[i],
\end{align*}
where the third equality holds since
\[
\sum_{j \in \pi_k^{-1}(m)} M_{\tau_k}[i,j]
= \sum_{j \in \pi_k^{-1}(m)} |\tau_k(i)|_j
= |\tau(\pi_k(i))|_m
= M_\tau[\pi_k(i),m].\qedhere
\]
\end{proof}

\begin{proposition}
\label{pro: tau2 Pisot implies tau Pisot}
Assume that $\tau_2$ is ultimately Pisot such that its characteristic polynomial is $X^m \cdot P_\theta(X)$ for some Pisot number $\theta$ and some integer $m\ge 0$.
Then the characteristic polynomial of $\tau$ is of the form $X^\ell \cdot P_\theta(X)$ with $\ell \le m$.
In particular, $\tau$ is ultimately Pisot with the same Pisot root.
\end{proposition}
\begin{proof}
From~\cref{lem:subset of spectres}, each eigenvalue of $\tau$ is one of $\tau_2$, so the characteristic polynomial of $\tau$ can be written as $X^\ell R(X)$ for some integer $\ell\le m$ and some polynomial $R(X)$ for which $0$ is not one of its zeroes and that divides $P_\theta(X)$.
Since $\theta$ is Pisot, $P_\theta(X)$ is irreducible and so $R(X)=P_\theta(X)$.
\end{proof}




\begin{theorem}
\label{thm:bounded}
Let $\infw{x}$ be a fixed point of a primitive substitution $\tau$.
If $\tau_2$ is ultimately Pisot, then the $k$-abelian complexity $(\abk{k}{x}(n))_{n\ge 0}$ is bounded for each $k\ge 1$.
\end{theorem}
\begin{proof}
From~\cite[Theorem~22]{Adamczewski:2003} (also see the beginning of~\cite[Section~6]{Adamczewski:2003}), the quantity $C_{\infw{x}}^k(n):=\max_{w\in \faset_k(\infw{x})} \max_{u,v\in \faset_n(\infw{x})} \{ ||u|_w-|v|_w| \}$ is bounded for all $k,n\ge 0$.
In particular, $(\abeqk{k}{x}(n))_{n\ge 0}$ is bounded.
Due to Item 2 of~\cref{lem:properties and observations for k-ab complexities}, $(\abk{k}{x}(n))_{n\ge 0}$ is also bounded.
\end{proof}


\newcommand{\ppref}{\operatorname{prefb}}
\newcommand{\pfac}{\operatorname{facb}}
\newcommand{\pmin}{\operatorname{minb}}
\newcommand{\pdiff}{\operatorname{diffb}}
\newcommand{\pborder}{\operatorname{border}}
\newcommand{\pconv}{\operatorname{conv}}

If we show that the $k$-abelian complexity is furthermore $\mathcal{N}_\tau$-regular, it is then $\mathcal{N}_\tau$-automatic. Recall $\pfeq_\infw{x}$ and $\pabeq_\infw{x}$ from~\cref{sec:method 1}.
We now introduce the following relations and functions:
\begin{align*}
    \ppref_{\infw{x},a}(n) &= \Psi\left(B_k(\infw{x})\interv{0}{n}\right)[a],\\
    \pfac_{\infw{x},a}(i,n) &= \ppref_{\infw{x},a}(i+n) - \ppref_{\infw{x},a}(n),\\
    \pmin_{\infw{x},a}(n) &= \min_{i\geq 0} \{ \pfac_{\infw{x},a}(i,n) \}, \\
    \pdiff_{\infw{x},a}(i,n) &=\pfac_{\infw{x},a}(i,n)-\pmin_{\infw{x},a}(n), \\
    \pborder_\infw{x} &= \left\{ (i,j,k,n) \mid \left(k\leq n \Rightarrow \pfeq_\infw{x}(i,j,k-1)\right) \land \left(n<k \Rightarrow \pfeq_\infw{x}(i,j,n)\right)\right\}.
\end{align*}

\begin{theorem}
\label{thm:regular}
    Let $k\ge 1$ be an integer and let $\infw{x}\in A^\NN$ be a sequence such that $B_k(\infw{x})$ is automatic in some regular numeration system $\mathcal{N}_k$.
    If the Parikh vectors of length-$n$ prefixes of $B_k(\infw{x})$ form an $\mathcal{N}_k$-synchronized sequence, then $(\abk{k}{x}(n))_{n\ge 0}$ is $\mathcal{N}_k$-regular.
\end{theorem}

\begin{proof}
Fix $k\ge 1$ and let $A_k=\{1,\ldots,m:=\facc{x}(k)\}$ be the alphabet over which $B_k(\infw{x})$ is defined. By hypothesis, the functions $\ppref_{\infw{x},a}$ are all synchronized. Therefore, the functions $\pfac_{\infw{x},a}$ are all synchronized with the following formula in first-order logic: $\pfac_a(i,n,z)=\exists x,y \; \ppref_{\infw{x},a}(i,x) \land \ppref_{\infw{x},a}(i+n,y) \land z+x=y$. Similarly, the functions $\pmin_{\infw{x},a}$ and $\pdiff_{\infw{x},a}$ are also all synchronized. We now provide a formula for $\pabeq_\infw{x}(i,j,k,n)$. This formula is split in two cases: if $n\leq k$, then the $k$-abelian equivalence is simply the factor equivalence, and otherwise we use the border condition; see \cref{lem:equivalent def for k-ab equiv}: \begin{align*}
\pabeq_\infw{x}(i,j,k,n)&= [n<k \land \pfeq_{\infw{x}}(i,j,n)] \land [ \pborder_\infw{x}(i,j,k,n+k) \\
 & \land \exists z (\pdiff_{\infw{x},1}(i,n,z) \land \pdiff_{\infw{x},1}(j,n,z) ) \land \cdots \\ & \land \exists z (\pdiff_{\infw{x},m}(i,n,z) \land \pdiff_{\infw{x},m}(j,n,z) ].
\end{align*} Therefore, this relation is also synchronized. Finally, we define the following relation that identifies the first occurrences of $k$-abelian equivalent factors: \begin{align*}
    \Lambda_\infw{x}= \left\{ (i,k,n) \mid \forall j \; \pabeq_\infw{x}(i,j,k,n) \implies i\leq j \right\}.
\end{align*} 
By using the path-counting technique \cite[Section 9.8]{Shallit:2023}, $(\abk{k}{x}(n))_{n\ge 0}$ is a regular sequence. 
\end{proof}

\begin{corollary}\label{cor:kab-auto}
Let $k\ge 1$ be an integer and let $\infw{x}$ be a fixed point of a primitive substitution $\tau$.
Let $\mathcal{N}_k$ be the numeration system associated with $\tau_k$.
If $\tau_2$ is ultimately Pisot, then the $k$-abelian complexity $(\abk{k}{x}(n))_{n\ge 0}$ is $\mathcal{N}_\tau$-automatic.
\end{corollary}

\begin{proof}
Note that $B_k(\infw{x})$ is the fixed point of $\tau_k$, so $B_k(\infw{x})$ is $\mathcal{N}_k$-automatic. Since a regular sequence that is bounded is automatic, we deduce that $(\abk{k}{x}(n))_{n\ge 0}$ is $\mathcal{N}_k$-automatic, by combining both \cref{thm:bounded,thm:regular}. Finally, we apply \cref{lem:subset of spectres} and \cref{prop:converter} to deduce that this sequence is also $\mathcal{N}_\tau$-automatic. 
\end{proof}

\begin{remark}
Let us notice that the sliding-block code of Parikh-collinear substitutions are not necessarily ultimately Pisot. For instance, resuming~\cref{ex: z1 Parikh-collinear}, the substitution for the length-$2$ sliding-block code is defined by $1 \mapsto 123144, 2 \mapsto 2312, 3 \mapsto 123142, 4 \mapsto 2314$ with polynomial $X^2(X-1)(X-5)$, which is therefore not ultimately Pisot (nor Pisot).  
\end{remark}

\subsection{Application to the Narayana sequence}
\label{sec:Narayana sequence}

We consider the sequence $\infw{n}=01200101201200120010\cdots$, fixed point starting with $0$ of the substitution $\tau: 0\mapsto 01, 1\mapsto 2, 2 \mapsto 0$.
Up to a renaming of the letters, it is the sequence~\cite[\seqnum{A105083}]{Sloane}.
This sequence $\infw{n}$ is called the \emph{Narayana sequence}; see~\cite{Shallit:2025}.
We note that the characteristic polynomial of the substitution $\tau$ is given by the minimal polynomial $P_\theta(X)=X^3-X^2-1$ of the Pisot number $\theta \approx 1.46557$. 

Many combinatorial properties of $\infw{n}$ have recently been studied by Shallit~\cite{Shallit:2025} using \texttt{Walnut} and by Letouzey~\cite{Letouzey:2025}.
For example, its factor complexity satisfies $\facc{s}(n)=2n+1$ (see~\cite[Theorem 13]{Shallit:2025}). As in the previous section, we let $\tau_k$ denote the substitution that generates the length-$k$ sliding-block code of $\infw{n}$. Since $\tau_2\colon 1 \mapsto 12, 2 \mapsto 3, 3 \mapsto 4, 4 \mapsto 15, 5 \mapsto 3$ is ultimately Pisot, with polynomial $X^2P_\theta(X)$, we can apply \Cref{cor:kab-auto}. Therefore, we have computed the $k$-abelian complexity of the sequence $\infw{n}$, up to $k=10$. The details of the \texttt{Walnut} implementation are provided in~\Cref{sec:appendix-Walnut}.

\begin{theorem}
Let $\infw{n}$ be the Narayana sequence, fixed point of $0\mapsto 01, 1\mapsto 2, 2 \mapsto 0$.
For $k\in[1,10]$, the $k$-abelian complexity of $\infw{n}$ takes on the values in the set given in~\cref{tab:first values of first few abelian complexities Narayana}.

\end{theorem}

\begin{table}[h!tbp]
       \centering
        \[
        \begin{array}{c|l|l}
            k & \{\abk{k}{n}(n) \mid n\ge 0\} & \text{Size of the automaton} \\
            \hline
            1 & \{1\} \cup [3, 8] & 97 \\
            2 & \{1, 3, 5, 7\} \cup[9, 22] & 277 \\
            3 & \{1, 3, 5, 7, 9, 11, 13\} \cup [15, 37] & 467 \\
            4 & \{1, 3, 5, 7, 9, 11, 13, 15, 17, 19\} \cup [21, 52] & 634 \\
            5 & \{2n+1 \mid 0\leq n \leq 11 \} \cup \{25, 26, 28, 29, 30, 31, 32\} \cup [34, 66] & 871 \\
            6 & \{2n+1 \mid 0\leq n \leq 16\} \cup [34, 81] & 969 \\
            7 & \{2n+1 \mid 0\leq n \leq 18\} \cup \{38\} \cup [40,47] \cup [49, 96] & 1218 \\
            8 & \{2n+1 \mid 0\leq n \leq 21\} \cup [47, 111] & 1309 \\
            9 & \{2n+1 \mid 0\leq n \leq 23\} \cup [49,52] \cup [54,63] \cup [65, 125] & 1646 \\
            10 & \{2n+1 \mid 0\leq n \leq 28\} \cup \{54\} \cup [59,68]  \cup \{70\} \cup [72, 140] & 1745 \\
        \end{array}
        \]
        \caption{For $k\in[1,10]$, the values taken by the $k$-abelian complexity of the Narayana sequence $\infw{n}$.}
        \label{tab:first values of first few abelian complexities Narayana}
    \end{table}

\subsection{Application to other sequences}

The second approach can also be applied to the four sequences studied with the first approach presented in \cref{sec:method 1}. In general, a good rule would be to first try the first approach and turn to the second approach if the computation does not converge in reasonable time (either because the sequence is not uniformly-factor-balanced or because it is too heavy). The supplementary material of the paper provides the reader with the results for all the sequences listed above plus the following: \begin{itemize}
     \item The fixed point of $0\mapsto 011,1\mapsto 01$, with Pisot root $1+\sqrt{2}$;
     \item The fixed point of $0\mapsto 0001011,1\mapsto 001011$, with Pisot root $\frac{7+\sqrt{37}}{2}$;
     \item The fixed point of $0\mapsto 001, 1\mapsto 02,2\mapsto 002$, with Pisot root of $X^3-3X^2+X-1$; 
     \item The fixed point of $0\mapsto 010,1\mapsto 2,2\mapsto 02$, with Pisot root of $X^3-3X^2+2X-1$;
     \item The \emph{twisted Tribonacci sequence}~\cite[\seqnum{A277735}]{Sloane}, fixed point of $0\mapsto 01,1\mapsto 20,2\mapsto 0$, with Pisot root of $X^3-X^2-X-1$. 
\end{itemize}

\section{Open problems and questions}
\label{sec:open pbs}

\begin{conjecture}
Let $\infw{t}$ be the Tribonacci sequence, fixed point of $0\mapsto 01, 1\mapsto 02, 2 \mapsto 0$.
The $2$-dimensional sequence $(\abk{k}{t}(n))_{k\ge 1, n\ge 0}$ is not synchronized but computed by a sequence automaton of polynomial $(X-1)(X^3-X^2-X-1)$.
\end{conjecture}

\begin{question}
 We checked the computed complexities for the Fibonacci sequence and the $3$-abelian complexity of the Tribonacci sequence.
 In general, can we obtain some inductive procedure to check/certify/validate our result?
\end{question}

Given in~\cref{thm:Fib k2balanced,thm:Trib k2balanced} on the Fibonacci and Tribonacci sequences, we raise the following question. 

\begin{question}
For an integer $m\ge 2$, let $\infw{x}_m$ be the $m$-bonacci sequence, fixed point of $0\mapsto 01, 1 \mapsto 02, \ldots, m-2 \mapsto 0(m-1), m-1 \mapsto 0$.
    What is the value of the smallest integer $C^{(m)}_k\ge 1$ such that $\infw{x}_m$ is $(k,C^{(m)}_k)$-balanced? Bounds on $C^{(m)}_1$ are given in~\cite{Brinda:2014}.
\end{question}


\begin{question}
 Let $\infw{x}$ be a fixed point of a substitution $\tau$.
 Consider a substitution $\sigma \colon A \to A^*$ that might be erasing and let $\infw{y} = \sigma (\infw{x})$.
 If $\tau$ is Pisot, then both $\infw{x}$ and $\infw{y}$ have automatic abelian complexities.
 Can we generalize this result to all $k$-abelian complexities? How do we compute the length-$k$ sliding-block code of the composition $\sigma \circ \tau$?
\end{question}



\section*{Acknolwedgments}
Pierre Popoli is supported by ULiège's Special Funds for Research, IPD-STEMA Program.
Jeffrey Shallit is supported by NSERC grant 2024-03725..
Manon Stipulanti is an FNRS Research Associate supported by the Research grant 1.C.104.24F

The authors want to thank Bastián Espinoza and Julien Leroy for useful discussions and the LIFO research laboratory for providing computation time on its \texttt{mirev} cluster.


\bibliographystyle{plainurl}
\bibliography{abco.bib}

\appendix

\section{Validation script for \cref{subsec:check validity}}

Here is the details of the \texttt{Walnut} script used to check the DFAO and ensure it computes $\pbal_\infw{t}$, where $\infw{t}$ is the Tribonacci sequence, fixed point of $0\mapsto 01, 1\mapsto 02, 2 \mapsto 0$.

\begin{lstlisting}
eval init "?msd_tri Ai,j1,j2,k Dequitri[i][j1][j2][k][0]=@-1
| Dequitri[i][j1][j2][k][0]=@0
| Dequitri[i][j1][j2][k][0]=@1":
\end{lstlisting}

\begin{lstlisting}
eval initXX "?msd_tri Ai,j1,j2,k 
        ($feq_tri(i,j1,k) <=> $feq_tri(i,j2,k)) 
    <=> Dequitri[i][j1][j2][k][0]=@0":
eval initTF "?msd_tri Ai,j1,j2,k 
        ($feq_tri(i,j1,k) & ~$feq_tri(i,j2,k)) 
    <=> Dequitri[i][j1][j2][k][0]=@1":
eval initFT "?msd_tri Ai,j1,j2,k 
        (~$feq_tri(i,j1,k) & $feq_tri(i,j2,k)) 
    <=> Dequitri[i][j1][j2][k][0]=@-1":
\end{lstlisting}

\begin{lstlisting}
def increase "?msd_tri 
  (Dequitri[i][j1][j2][k][n]=@-2 & Dequitri[i][j1][j2][k][n+1]=@-1) 
| (Dequitri[i][j1][j2][k][n]=@-1 & Dequitri[i][j1][j2][k][n+1]=@0) 
| (Dequitri[i][j1][j2][k][n]=@0 & Dequitri[i][j1][j2][k][n+1]=@1)
| (Dequitri[i][j1][j2][k][n]=@1 & Dequitri[i][j1][j2][k][n+1]=@2)":
def decrease "?msd_tri 
  (Dequitri[i][j1][j2][k][n]=@-1 & Dequitri[i][j1][j2][k][n+1]=@-2)
| (Dequitri[i][j1][j2][k][n]=@0 & Dequitri[i][j1][j2][k][n+1]=@-1)
| (Dequitri[i][j1][j2][k][n]=@1 & Dequitri[i][j1][j2][k][n+1]=@0) 
| (Dequitri[i][j1][j2][k][n]=@2 & Dequitri[i][j1][j2][k][n+1]=@1)":
def constant "?msd_tri 
  (Dequitri[i][j1][j2][k][n]=@-2 & Dequitri[i][j1][j2][k][n+1]=@-2)
| (Dequitri[i][j1][j2][k][n]=@-1 & Dequitri[i][j1][j2][k][n+1]=@-1)
| (Dequitri[i][j1][j2][k][n]=@0 & Dequitri[i][j1][j2][k][n+1]=@0)
| (Dequitri[i][j1][j2][k][n]=@1 & Dequitri[i][j1][j2][k][n+1]=@1) 
| (Dequitri[i][j1][j2][k][n]=@2 & Dequitri[i][j1][j2][k][n+1]=@2)":
\end{lstlisting}

\begin{lstlisting}
eval nxt "?msd_tri Ai,j1,j2,k,n
    $constant(i,j1,j2,k,n) 
  | $increase(i,j1,j2,k,n) 
  | $decrease(i,j1,j2,k,n)":
\end{lstlisting}

\begin{lstlisting}
eval nxtXX "?msd_tri Ai,j1,j2,k,n 
        ($feq_tri(i,j1+n+1,k) <=> $feq_tri(i,j2+n+1,k)) 
    <=> $constant(i,j1,j2,k,n)":
eval nxtTF "?msd_tri Ai,j1,j2,k,n
        ($feq_tri(i,j1+n+1,k) & ~$feq_tri(i,j2+n+1,k))
    <=> $increase(i,j1,j2,k,n)":
eval nxtFT "?msd_tri Ai,j1,j2,k,n
        (~$feq_tri(i,j1+n+1,k) & $feq_tri(i,j2+n+1,k))
    <=> $decrease(i,j1,j2,k,n)":
\end{lstlisting}

\section{\texttt{Walnut} code for~\cref{sec:Fibonacci,sec:Pell seq,sec:Tribonacci}}
\label{sec:appendix-C}

\cref{thm:Fib k2balanced} can be proven by running the following \texttt{Walnut} code, returning \texttt{TRUE}:

\begin{lstlisting}
eval b1fib "?msd_fib Ai,j1,j2,k,n Dequifib[i][j1][j2][k][n] <= @2":
\end{lstlisting}

\cref{thm: balanced Fib} can be proven by running the following \texttt{Walnut} code, where the last command returns \texttt{TRUE}:

\begin{lstlisting}
def unb1fib "?msd_fib Ai Ej1,j2,n Dequifib[i][j1][j2][k][n] > @1":
eval allfrom4 "?msd_fib Ak $unb1fib(k) <=> k>=4":
\end{lstlisting}

Similarly, we run the following command to obtain a proof of~\cref{thm:Pell k3balanced}, returning \texttt{TRUE}:

\begin{lstlisting}
eval b1pell "?msd_pell Ai,j1,j2,k,n Dequipell[i][j1][j2][k][n]<= @3":
\end{lstlisting}

The following commands give a proof of the first part of~\cref{thm: unbalanced Pell}, returning \texttt{TRUE}:

\begin{lstlisting}
def unb1pell "?msd_pell Ai Ej1,j2,n Dequipell[i][j1][j2][k][n] > @1":
eval allfrom6 "?msd_pell Ak $unb1pell(k) <=> k>=6":
\end{lstlisting}

And the following commands proof of the second part of~\cref{thm: unbalanced Pell}, returning an automaton recognizing the empty set:

\begin{lstlisting}
def unb2pell "?msd_pell Ai Ej2,j2,n Dequipell[i][j2][j2][k][n] > @2":
eval unb2 "?msd_pell $unb2pell(k)":
\end{lstlisting}

\cref{thm:Trib k2balanced} can be proven by running the following \texttt{Walnut} code, returning \texttt{TRUE}:

\begin{lstlisting}
eval b1tri "?msd_tri Ai,j1,j2,k,n Dequitri[i][j1][j2][k][n] <= @2":
\end{lstlisting}

\cref{thm:Trib k2unbalanced} can be proven by running the following \texttt{Walnut} code, where the last command returns \texttt{TRUE}:

\begin{lstlisting}
def to2tri "?msd_tri Dequitri[i][j1][j2][k][n] > @1":
def tri2tri "?msd_tri Ej1,j2 $to2tri(i,j1,j2,k,n)":
def unb1tri "?msd_tri Ai En $tri2tri(i,k,n)":
eval allfrom "?msd_tri Ak $unb1tri(k) <=> k>=1":
\end{lstlisting}

\section{Exemplification of~\cref{sec:k-abelian complexity} with the Thue--Morse sequence}
\label{sec:appendix-B}

We illustrate the concepts of the beginning of~\cref{sec:k-abelian complexity}.
For this, we fix the Thue--Morse sequence $\infw{x}=01 10 10 01 10 01 01 10\cdots$, fixed point of the substitution $\tau: 0\mapsto 01, 1\mapsto 10$.
We start with its sliding-block code of length $2$ from~\cref{def: sliding block code}.
Since $\#\faset_2(\infw{f})=4$, $B_2(\infw{x})$ is encoded over an alphabet of four letters and we have $B_2(\infw{x})=12 31 34 12 34 13\cdots$.
Let us obtain the substitution generating this latter sequence.
We get $A_2=\{1,2,3,4\}$ and the encoding $\Theta_2 \colon 01\mapsto 1, 11 \mapsto 2, 10  \mapsto 3, 00 \mapsto 4.$
Now $\tau_2$ is defined by $1\mapsto 12$, $2 \mapsto 31$, $3 \mapsto 34$, and $4 \mapsto 13$.
For instance, since $\ell=1$ encodes the factor $u=01$ of $\infw{x}$ and $\tau(u[0])=\tau(0)=01$ has length $2$, we look at the first $2$ length-$2$ factors of $\tau(u)=0110$, i.e.,  $\tau_2(1)=\underbrace{1}_{=01}\underbrace{2}_{=11}$.

We now illustrate the proof of~\cref{lem:subset of spectres}. The eigenvalues of $M_{\tau}=\begin{pmatrix}
    1 & 1 \\
    1 & 1
\end{pmatrix}$ are $0$ and $2$ and respective eigenvectors are given by \[
V_0 = \begin{pmatrix} 1 \\ -1 \end{pmatrix} 
\quad \text{and} \quad
V_2 = \begin{pmatrix} 1 \\ 1 \end{pmatrix}.
\]
Now seeing $\Theta_2$, we obtain that $\pi_2 \colon 1 \mapsto 0, 2 \mapsto 1, 3 \mapsto 1, 4 \mapsto 0$.
For instance, since $\Theta_2^{-1}(1)=12$, we look at the first letter of the factor $0110$ coded by $12$ of $\infw{x}$ to obtain $\pi_2(1)$, which is $0$.
So the vectors
\[
V'_0 = \begin{pmatrix} V_0[\pi_2(1)] \\ V_0[\pi_2(2)] \\ V_0[\pi_2(3)] \\ V_0[\pi_2(4)] \end{pmatrix}
=  \begin{pmatrix} 1 \\ -1 \\ -1 \\ 1 \end{pmatrix}
\quad \text{and} \quad
V'_2 = \begin{pmatrix} V_2[\pi_2(1)] \\ V_2[\pi_2(2)] \\ V_2[\pi_2(3)] \\ V_2[\pi_2(4)] \end{pmatrix} 
= \begin{pmatrix} 1 \\ 1 \\ 1 \\ 1 \end{pmatrix} 
\]
are eigenvectors of $M_{\tau_2}$ with respective eigenvalues $0$ and $2$.

We now observe that the converse of~\cref{pro: tau2 Pisot implies tau Pisot} does not hold.
The Thue--Morse substitution $\tau\colon 0\mapsto 01, 1\mapsto 10$ is ultimately Pisot with characteristic polynomial $P_\tau(X)=X(X-2)$.
However, $\tau_2$ is not, since it has characteristic polynomial $P_{\tau_2}(X)=X(X-1)(X+1)(X-2)$.

\section{\texttt{Walnut} details for the Narayana sequence }\label{sec:appendix-Walnut}

In this section, we illustrate the method of \Cref{sec:k-abelian complexity} on the Narayana sequence, fixed point of $\tau\colon 0\mapsto 01, 1\mapsto 2, 2 \mapsto 0$.
In \texttt{Walnut}, $\infw{n}$ is encoded by \verb|Nara| and we also define the corresponding Dumont--Thomas numeration system; see \cref{list:deftri}.
\begin{lstlisting}[caption={Generate the Dumont--Thomas numeration system for the Narayana substitution.},label=list:deftri,float=h]
%%python
from licofage.kit import *
import os
setparams(True, True, os.environ["WALNUT_HOME"])

s = subst('01/2/0')
ns = address(s, "nara")
ns.gen_ns()
ns.gen_word_automaton()
\end{lstlisting}

Then we set a factor comparison predicate in \texttt{Walnut} and a first factor occurrence predicate; see \cref{list:factor}.

\begin{lstlisting}[caption={The predicates for factor comparison and first factor occurrence.},label=list:factor,float=h]
def cut "?msd_nara i<=u & j<=v & u+j=v+i & u<n+i & v<n+j":
def feq_nara "?msd_nara ~(Eu,v $cut(i,j,n,u,v) & Nara[u]!=Nara[v])":
eval comp_nara n "?msd_nara Aj $feq_nara(i,j,n) => i<=j":

\end{lstlisting}

As explained in~\cite[Section~8.1]{Fici:2023}, we use these predicates to define the border condition of~\cref{lem:equivalent def for k-ab equiv}:
\begin{lstlisting}
def bordercond "?msd_nara (k<=n => $feq_nara(i,j,k-1)) 
                    & (n<k => $feq_nara(i,j,n))":
\end{lstlisting}

Let us agree that we want to compute the $3$-abelian complexity $(\abk{3}{n}(n))_{n\ge 0}$ of $\infw{n}$.
Then we need the length-$3$ sliding-block code $B_3(\infw{n})$ of $\infw{n}$ and the corresponding substitution $\tau_3$.
By~\cite[Theorem 13]{Shallit:2025}, the sequence $\infw{n}$ have $2\cdot 3+1=7$ length-$3$ factors.
The substitution $\tau_3$ is thus over $7$ letters, and one can check that \[
    \tau_3 \colon 0\mapsto 01, 1 \mapsto 2, 2 \mapsto 3, 3 \mapsto 45, 4 \mapsto 12, 5 \mapsto 6, 6 \mapsto 3. 
    \]

We then obtain the Dumont--Thomas numeration systems associated with both $\tau,\tau_3$ and convert one to another.
We observe that the conversion is the identity (as can also be seen in~\cref{fig:conv_nara-nara3}), but this is not always the case.
We use the following code to obtain our results:
\begin{lstlisting}
%%python
s3 = block(s, 3)
ns3 = address(s3, "narab3")
ns3.gen_ns()
(ns-ns3).gen_dfa("conv_nara_narab3")
\end{lstlisting}

\begin{figure}[h!tbp]
        \centering
        \includegraphics[width=0.5\linewidth]{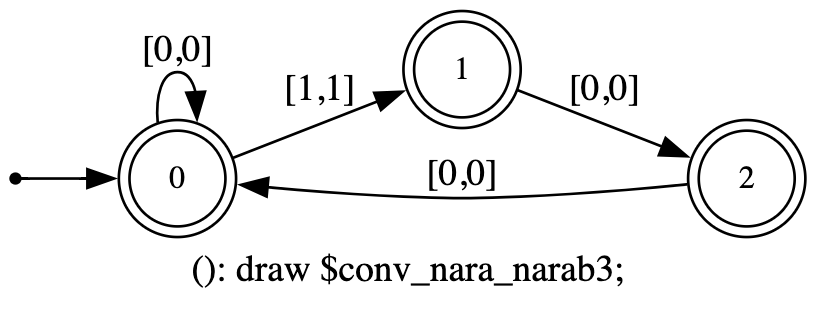}
        \caption{The converter between the Dumont--Thomas numeration systems associated with the Narayana substitution $\tau$ and the substitution behind the length-$3$ sliding-block code of its fixed point (here it computes the identity).}
        \label{fig:conv_nara-nara3}
\end{figure}

We translate the border condition into the new numeration system:
\begin{lstlisting}
def bordercond3 "?msd_narab3 (?msd_nara Eii,jj,kk,nn (
                    $conv_nara_narab3(?msd_nara ii, ?msd_narab3 i) & 
                    $conv_nara_narab3(?msd_nara jj, ?msd_narab3 j) & 
                    $conv_nara_narab3(?msd_nara kk, ?msd_narab3 k) & 
                    $conv_nara_narab3(?msd_nara nn, ?msd_narab3 n) & 
                    $bordercond(ii,jj,kk,nn)))":
\end{lstlisting}
    
We are now able to compute $(\abk{3}{n}(n))_{n\ge 0}$. We first compute the Parikh vectors for the prefixes of $B_3(\infw{n})$:

\begin{lstlisting}
%%python
for (m,a) in enumerate(ns3.alpha):
    w = {'_': 0}
    w[a] = 1
    parikh = address(s3, ns3.ns, **w)
    (parikh - ns3).gen_dfa(f"narab3p{m}")
\end{lstlisting}

Second, for each $i\in[0,6]$ (since $B_3(\infw{n})$ is over $7$ letters), we write the following predicates:

\begin{lstlisting}
def fac{m} "?msd_narab3 Ex,y $narab3p{m}(i,x) & $narab3p{m}(i+n,y) 
                                          & z+x=y":
def min{m} "?msd_narab3 Ei $fac{m}(i,n,x) & Aj,y $fac{m}(j,n,y) 
                                          => y>=x":
def diff{m} "?msd_narab3 Ex,y $min{m}(n,x) & $fac{m}(i,n,y) & z+x=y":
\end{lstlisting}

For instance, \verb|fac0(i,n,z)| insures that $\verb|z|$ gives the number of letters $0$ in the length-$n$ factor $B_3(\infw{x})\interv{i}{i+n}$. Similarly, \verb|min0(n,x)| insures that $\verb|x|$ is the smallest number of $0$'s in all length-$n$ factors of $B_3(\infw{x})$; and \verb|diff0(i,n,z)| insures that $\verb|z|$ is the quantity needed to obtain the number of $0$'s in $B_3(\infw{x})\interv{i}{i+n}$ from the minimum number of $0$'s in all length-$n$ factors of $B_3(\infw{x})$. Then we combine all $7$ predicates to obtain the $3$-abelian complexity as follows:
\begin{lstlisting}
def abeq_narab3 "?msd_narab3 $bordercond3(i,j,3,n+2) 
 & (Ez $diff0(i,n,z) &  $diff0(j,n,z)) 
 & (Ez $diff1(i,n,z) &  $diff1(j,n,z)) 
 & (Ez $diff2(i,n,z) &  $diff2(j,n,z)) 
 & (Ez $diff3(i,n,z) &  $diff3(j,n,z)) 
 & (Ez $diff4(i,n,z) &  $diff4(j,n,z)) 
 & (Ez $diff5(i,n,z) &  $diff5(j,n,z)) 
 & (Ez $diff6(i,n,z) &  $diff6(j,n,z))":
\end{lstlisting}

Finally, to get back to the original numeration system and to compute the first values, we use the following predicate:
\begin{lstlisting}
def abeq_nara3 "?msd_nara (n<2 & $feq_nara(i,j,n)) 
      | (n>=2 & (?msd_narab3 Ei,j,n 
      ($conv_nara_narab3(?msd_nara ii, ?msd_narab3 i) 
     & $conv_nara_narab3(?msd_nara jj, ?msd_narab3 j) 
     & $conv_nara_narab3(?msd_nara nn, ?msd_narab3 n) 
     & $abeq_narab3(ii,jj,nn-2))))":
\end{lstlisting}

The following command gives a linear representation of the $3$-abelian complexity: 
\begin{lstlisting}
eval comp_nara3 n "?msd_nara Aj $abeq_nara3(i,j,n) => i<=j":
\end{lstlisting}

Finally, by applying the semigroup trick, we obtain the desired DFAO for the $3$-abelian complexity: 
\begin{lstlisting}
%SGT comp_nara3 msd_nara Comp_nara3
\end{lstlisting}

\end{document}